\documentclass[12pt,a4paper,final]{iopart}

\usepackage{iopams}  
\usepackage[dvips]{graphicx}
\usepackage[breaklinks=true,colorlinks=true,linkcolor=blue,urlcolor=blue,citecolor=blue]{hyperref}
\expandafter\let\csname equation*\endcsname\relax
\expandafter\let\csname endequation*\endcsname\relax
\usepackage{amsmath}
\usepackage{ulem}
\usepackage{cancel}
\usepackage{amsthm}
\usepackage{mathtools}
\usepackage{tikz}
\usepackage{graphicx}
\usepackage{mathrsfs}
\graphicspath{ {./images/} }
\newtheorem{theorem}{Theorem}[section]
\newtheorem{corollary}{Corollary}[theorem]
\newtheorem{lemma}[theorem]{Lemma}
\newcommand{\vct}{\textbf}
\newcommand{\spc}{\qquad}
\newcommand{\npg}{\vspace{\baselineskip}}
\newcommand{\eqn}{\npg\indent}
\newcommand{\npgni}{\npg\noindent}
\newcommand{\tph}{^\frac{1}{2}}

\begin{document}
\title{Newtonian Quantum Gravity, WIMPs and Celestial Mechanics-Summary of Results}

\author{Richard Durran, Andrew Neate, Aubrey Truman}
\address{Department of Mathematics, Computational Foundry, Swansea University Bay Campus, Fabian Way, Swansea, SA1 8EN, UK}
\ead{a.truman@swansea.ac.uk} 

\begin{abstract}

We discuss the leading term in the semi-classical asymptotics of Newtonian quantum gravity for the Kepler problem. For dark matter, ice or dust particles in the gravitational field of a star or massive planet this explains how rapidly planets or ring systems can be formed by considering semi-classical equations for the asymptotics of the relevant Schrödinger wave function, especially the atomic elliptic state of Lena, Delande and Gay. Our results in this connection are summarised in Theorem 3.2. We extend this treatment to isotropic harmonic oscillator potentials in two and three dimensions finding explicit solutions, with important applications to the Trojan asteroids. More detailed results on the restricted 3-body problem are in a forthcoming paper. Further in two dimensions, we find the explicit solutions for Newton's corresponding revolving orbits, explaining planetary perihelion advance in these terms. A general implicit approach to solving our equations is given using Cauchy characteristic curves giving necessary and sufficient conditions for existence and uniqueness of our solutions. Using this method we solve the two-dimensional Power Law Problem in our setting. Our methods give an insight to the behavior of semi-classical orbits in the neighborhood of classical orbits showing how complex constants of the motion emerge from the SO(4) symmetry group and what the particle densities have to be for circular spirals for different central potentials. Such constants have an important role in revolving orbits. Moreover, they and the aforementioned densities could give rise to observable effects in the early history of planetary ring systems. We believe the quantum spirals here associated with classical Keplerian elliptical orbits explain the evolution of galaxies, with a neutron star or black hole at their centre, from spiral to elliptical. In this connection we include a brief account of some recent results on quantum curvature, torsion as well as anti-gravity and spiral splitting, acid tests of our ideas.

\end{abstract}

\section{Introduction}

In our previous papers, we have seen how in Schrödinger quantum mechanics, if the stationary state wave function $\psi\sim\textrm{exp}\left(\frac{R+iS}{\epsilon^2}\right)$ as $\epsilon^2=\hbar\sim0$, for potential $V$, it is necessary that for real valued $R$ and $S$,\vspace{-5mm}

$$\boldsymbol\nabla R.\boldsymbol\nabla S=0,\spc2^{-1}(|\boldsymbol\nabla S|^2-|\boldsymbol\nabla R|^2)+V=E\;\;\textrm{and}\;\;S\rceil_{C_0}=S_0,\spc(\ast)$$\vspace{-10mm}

\npgni where $C_0$ is the corresponding classical orbit, $\boldsymbol\nabla R\rceil_{C_0}=\vct{0}$ and $S_0$ is the corresponding Hamilton-Jacobi function, $E$ being the energy.

\npgni These semi-classical equations and how to solve them will be the main concerns in the present paper. To the extent that for the Kepler-Coulomb problem they give the leading term in the semi-classical asymptotics of Newtonian quantum gravity, they are relevant in modelling the behaviour of ice, dust or dark matter particles in the early history of solar systems. In particular they show how quickly planets and ring systems can be formed. (Refs. [1],[5],[6],[7],[8],[15],[16],[17]). Further we argue that galactic evolution is in part governed by the first quantisation of the Kepler problem for elliptical orbits and our equations. We begin by outlining some of the main ideas and results in sections (2) and (3).

\npgni In spite of the non-linearity of these equations they admit a linear superposition of solutions $R$ and $S$. In our treatment the log density, $\left(\frac{2R}{\epsilon^2}\right)$, is the driver and superposing two states leads to the splitting of spirals in the corresponding WIMP nebula (see pp. 9, 48).

\npgni In this setting of a constrained Hamiltonian system, the orbit is given by $t\to \vct{X}^0_t, t\in\mathbb{R}_{+}$, being the time, $\vct{X}_t^0\in\mathbb{R}^d$, d dimensions of space, where

$$\frac{d}{dt}\vct{X}_t^0=(\boldsymbol\nabla S+\boldsymbol\nabla R)(\vct{X}_t^0).$$\vspace{-10mm}

\npgni Since\vspace{-5mm}

$$\frac{d}{dt}R(\vct{X}_t^0)=(\boldsymbol\nabla S+\boldsymbol\nabla R)(\vct{X}_t^0).\boldsymbol\nabla R(\vct{X}_t^0)=|\boldsymbol\nabla R|^2(\vct{X}_t^0)\ge 0,$$

\npgni $R$ is monotonic increasing in time as the orbit is traversed, so $R(\vct{X}_t^0)\nearrow R_{\textrm{max}}$, the maximum value of $R$. If, as you expect, this maximum value of $R$ is attained on $C_0$, the classical orbit, then in the infinite
time limit the orbit will converge to classical motion on $C_0$, the convergence coming from the Bohm potential $-|\boldsymbol\nabla R|^2$, since

$$2^{-1}(|\boldsymbol\nabla S|^2+|\boldsymbol\nabla R|^2)+V-|\boldsymbol\nabla R|^2=E.$$

\npgni A supreme example is provided by the atomic elliptic state, $\psi$, of Lena, Delande and Gay. $\psi$ is localised on a classical Keplerian ellipse (with force centre at one focus) and minimises the associated SO(4) angular momentum uncertainty relations.(Refs. [5],[9],[12],[15],[17]). In this case, the classical Keplerian ellipse is $C_0$ and in the infinite time limit the corresponding orbit is described according to Kepler’s laws for potential $V=-\dfrac{\mu}{r}$, $\mu$ gravitational mass at focus $O$, dimension $d$ being 3. So in this semi-classical quantum mechanics of Newtonian gravity, if we take as our Schrödinger ensemble a cloud of identically prepared ice, dust particles or dark matter WIMPs, we have a simple spiral model for the formation of planets or ring systems. Moreover, we can see how long it takes to form a ring by estimating $t$, the infimum of $s$, such that

$$R_{\textrm{max}}-R_0=\int^s_0|\boldsymbol\nabla R|^2(\vct{X}_u^0)du.$$\vspace{-8mm}

\npgni In such a model it is plausible e.g. that gas giants such as Jupiter could be formed in a few million years as required in current planet building scenarios. A quick application is provided by the Levi-Civita transformation whereby $(x+iy)\to(x+iy)^2$, accompanied by the time-change from physical time $t$ to a fictitious time $s$, give the elliptical orbit corresponding to motion on an ellipse (with force centre at the centre of the ellipse, not the focus) for potential $V=\omega^2r^2/2$. Such a potential would be relevant if the ensemble of matter was inside a stable, homogeneous cloud of gravitating particles. The details of this are given in the next section (all part of the miracle of the inverse square law of force). (Ref. [12]). In this context we treat the problem of the formation of the Trojan asteroids revealing a link with the isotropic oscillator and more generally two new constants of the motion for the linearised restricted 3-body problem.\vspace{-3mm}

\npgni Armed with these examples, the rest of the paper is to try to extend the class of potentials for which our equations are soluble. Here we focus on the 2-dimensional situation, for $V_0=-\dfrac{\mu}{r}$, and for $V=V_0+\dfrac{C}{r^2},\:r^2=x^2+y^2$. In the latter case we rediscover Newton’s revolving orbits by considering certain complex constants of the motion. As an application we discuss the advance of the perihelion of Mercury in a general relativistic potential giving the agreement with what is observed. This is the content of section (4).\vspace{-3mm}

\npgni In section (5) we tackle equations directly using Cauchy’s method of characteristics. As an illustrative example we solve the equations for the power law in 2-dimensions, $-\kappa r^p$, $p$ being the power and a generalised LCKS transformation in two-dimensions. (Ref. [14]). More importantly, we give a detailed set of results for the Keplerian motion in a neighbourhood of the classical orbit on $C_0$. Once again complex constants emerge quite naturally resonating with the methods for Newton’s revolving orbits. Necessary and sufficient conditions are given for uniqueness and existence of solutions by converting our partial differential equations to ordinary differential equations.\vspace{-3mm}

\npgni For the sake of completeness and to help those readers not so familiar with celestial mechanics we include here a simple diagram (Figure 1) explaining why for the classical trajectory, $t\to \vct{X}^0_t=\vct{r}(t)$, on the classical Keplerian ellipse with corresponding eccentric anomaly $v$,\vspace{-5mm}

$$\left(\frac{d^{3}}{dv^{3}}+\frac{d}{dv}\right)\vct{r}=0,$$\vspace{-5mm}

\npgni where $v$ satisfies the Kepler equation for physical time $t$,

$$v-e\sin{v}=\sqrt{\frac{\mu}{a^{3}}}t,$$\vspace{-7mm}

\npgni $e$ being the orbital eccentricity, $a$ the semi-major axis of the orbit and $\mu$ the gravitational mass at the force centre. Surprisingly this result requires no change in dependent coordinates as the Levi-Civita transformation demands and amounts to a mere time-change. Several authors on the internet refer to this and related results as deep mysteries of the Kepler problem. Our diagram and simple trigonometric identities dispel this mystery. Here $t$ is the physical time measured from the pericentre, the point of nearest approach of the orbit to the force centre, $\vct{i}=\hat{\vct{r}}(0)$, $\vct{j}=\hat{\dot{\vct{r}}}(0)$ and if $\vct{r}=\overrightarrow{\textrm{SP}}$,

$$\vct{r}=(a\cos{v}-ae)\vct{i}+\sqrt{1-e^2}a\sin{v}\vct{j}.$$\vspace{-7mm}

\npgni Setting $\vct{R}=\vct{r}+ae\vct{i}$ and $(')=\dfrac{d}{dv}$, we see that

$$\vct{R}''=-\vct{R}\;\;\; \textrm{and so}\;\;\; \vct{r}'''+\vct{r}'=0.$$\vspace{-10mm}

\npgni Since $\vct{R}(v)=\vct{R}(0)\cos{v}+\vct{R}'(0)\sin{v}$, we obtain

$$\vct{r}(v)=(\vct{r}(0)+ae\vct{i})\cos{v}+\vct{r}'(0)\sin{v}\vct{j}-ae\vct{i},$$\vspace{-10mm}

\npgni for $v\ge0$. These results are important in understanding our Keplerian coordinates and important in their own right.

\begin{center}
\includegraphics[scale=0.75]{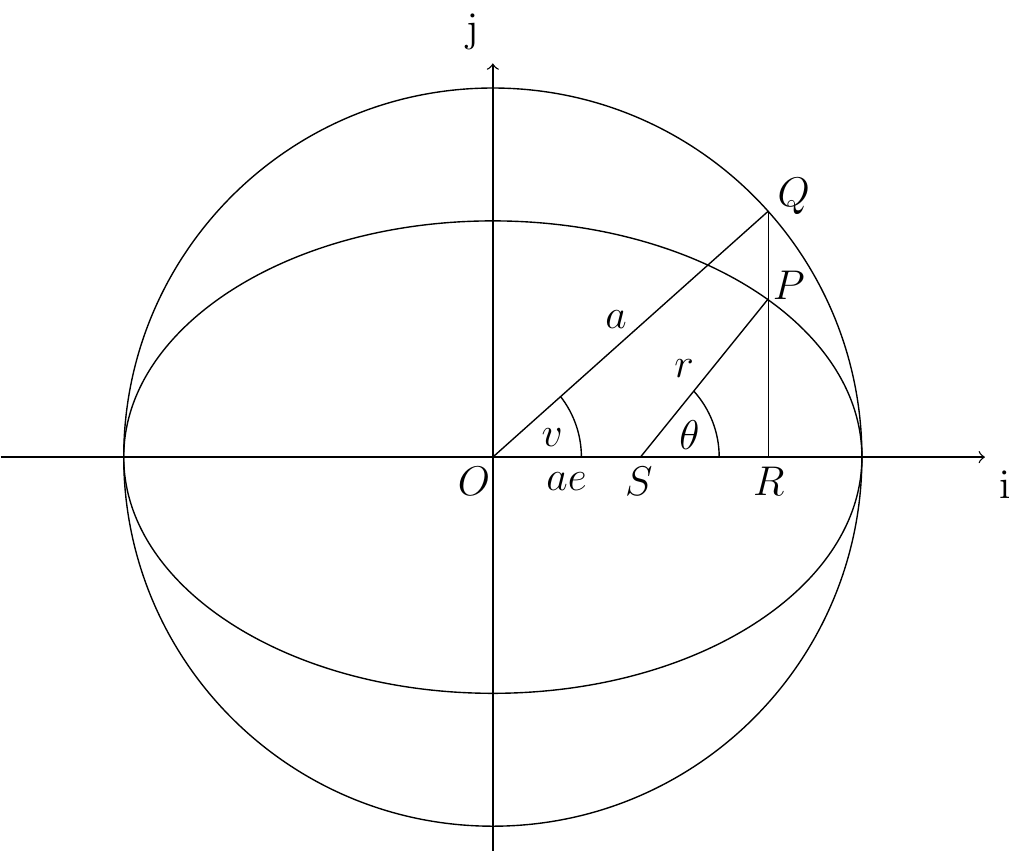}\\
Figure 1.
\end{center}\vspace{-3mm}

$$r=a(1-e\cos v),\spc\tan\dfrac{v}{2}=\sqrt{\dfrac{1-e}{1+e}}\tan\dfrac{\theta}{2},\spc PR=\sqrt{1-e^{2}}QR.$$

\section{Semi-classical Asymptotics of Newtonian Gravity and Celestial Mechanics}

\subsection{From Atomic to Astronomical Elliptic States}

In this section we use some standard tools from celestial mechanics to help to reveal some of the secrets of impenetrable algebraic expressions coming from the SO(4) symmetry group of the Kepler problem. Our basic ingredient is the beautiful stationary state of Lena, Delande and Gay, $\psi_{n,e}(\vct{r})$, the atomic elliptic state,

$$\psi_{n,e}(\vct{r})=C\textrm{exp}\left(-\frac{n\mu r}{\lambda^2}\right)L_{n-1}(n \nu(\vct{r})),$$\vspace{-5mm}

\npgni $L_n,n=0,1,2\dots$, being a Laguerre polynomial, $\vct{r}=(x,y,z)$, $r=|\vct{r}|=\sqrt{x^2+y^2+z^2}$,

$$\nu(\vct{r})=\frac{\mu}{\lambda^2}\left(r-\frac{x}{e}-\frac{iy\sqrt{1-e^2}}{e}\right).$$\vspace{-5mm}

\npgni Here our ellipse $\mathcal{E}_K$ has eccentricity $e$, $0<e<1$, with semi-major axis, $a=\dfrac{\lambda^2}{\mu}$, semi-latus rectum, $\ell=a(1-e^2)$, in the plane $z=0$, with semi-major axis parallel to the x-axis, with polar equation:

$$\frac{\ell}{r}=1+e\cos\theta,\spc\ell=\frac{h^2}{\mu},\spc e=\sqrt{1+\frac{2h^2E}{\mu^2}},\spc\textrm{the constants}$$

\npgni $E<0$ being the energy, $h$ the angular momentum about the z-axis and $\mu$ the gravitational mass at the focus of our ellipse $\mathcal{E}_K$. Most importantly, if $H$ denotes the quantum mechanical Hamiltonian, $H=2^{-1}|\vct{P}|^2-\dfrac{\mu}{|\vct{Q}|}$, then

$$H\psi_{n,e}=E_n\psi_{n,e},\spc E_n=-\frac{\mu^2}{2\hbar^2n^2},\spc\lambda=n\hbar,\spc \hbar=\epsilon^2.$$\vspace{-5mm}

\npgni We define $\vct{Z}_{n,e}(\vct{r})=\dfrac{\hbar}{i}\boldsymbol\nabla\ln\psi_{n,e}\to(\boldsymbol\nabla S-i\boldsymbol\nabla R)(\vct{r})\stackrel{\textrm{def}}{=}\vct{Z}(\vct{r})$. We take the Bohr correspondence limit $n\nearrow\infty$, $\epsilon\searrow0$, for fixed $\lambda$ i.e. fixed energy $E<0$. (Refs. [2],[4],[5],[17],[23]). In a previous paper we proved, in Cartesians, that\vspace{5mm}

$$\vct{Z}(\vct{r})=\frac{i\mu}{2\lambda}(1+\gamma)\frac{r}{|\vct{r}|}+\frac{\mu}{2\lambda e}(1-\gamma)\left(i,\;-\sqrt{1-e^2},0\right),$$\vspace{-10mm}

\npgni where\vspace{-5mm}

$$\gamma(\vct{r})=1-\lim_{n\nearrow\infty,\;\epsilon\searrow0}\frac{L'_{n-1}(n\nu(\vct{r}))}{L_n(n\nu(\vct{r}))}=\sqrt{1-\frac{4}{\nu(\vct{r})}}.$$\vspace{-5mm}

\npgni This gives the semi-classical state $\psi_{\textrm{sc}}$,

$$\psi_{\textrm{sc}}(\vct{r})=\nu(\vct{r})^\frac{\lambda}{\hbar}(1+\gamma(\vct r))^\frac{2\lambda}{\hbar}\textrm{exp}\left(-\frac{\mu}{2\lambda\hbar}|\vct{r}|+\frac{\lambda}{2\hbar}(1-\gamma(\vct{r}))\nu(\vct{r})\right),$$\vspace{-5mm}

\npgni so to within an additive constant, determined by Laplace's principle and the usual normalisation,\vspace{-5mm}

$$(R+iS)(\vct{r})=-\frac{\mu}{\lambda}r+\frac{\lambda \nu(\vct r)}{2}\left(1-\sqrt{1-\frac{4}{\nu(\vct r)}}\right)-\lambda\ln \nu(\vct r)-2\lambda\ln\left(1-\sqrt{1-\frac{4}{\nu(\vct r)}}\right),$$\vspace{-10mm}

\npgni $\nu$\ above.

\npgni (Refs. [5],[6],[7],[8],[15],[16],[17]).

\npgni It is difficult to see any connection with the last equation and the Kepler ellipse detailed above. The key is the eccentric anomaly $v$ for a point $P$ on an ellipse.

\npgni\textbf{Definition}\vspace{-3mm}

\npgni For an ellipse $\mathcal{E}$, with major axis parallel to the x-axis, the eccentric anomaly, $v$, of a point $P$ on $\mathcal{E}$ is the polar angle measured at the centre of $\mathcal{E}$ of the point $Q$ on the circumcircle of $\mathcal{E}$, $Q$ being the image under the parallel projection to the y-axis of $P$ on $\mathcal{E}$ i.e. for cartesian coordinates $(x,y)$ with origin at the right hand focus,\vspace{-2mm}

$$x=a(\cos v-e),\spc y=a\sqrt{1-e^2}\sin v,\spc r=\sqrt{x^2+y^2}=a(1-e\cos v),$$

$$\tan\frac{v}{2}=\sqrt{\frac{1-e}{1+e}}\tan\frac{\theta}{2},$$

\npgni $(r,\theta)$ polar coordinates of $P$ on $\mathcal{E}$ measured from the focus.

\subsection{Remarkable Formulae on}\vspace{-8.5mm}

\hspace{54mm} $\mathcal{E}_K$

$$\nu=-\frac{(1-e\textrm{e}^{iv})^2}{e\textrm{e}^{iv}},\spc\gamma=\sqrt{1-\frac{4}{\nu}}=\frac{(1+e\textrm{e}^{iv})}{(1-e\textrm{e}^{iv})}=\alpha+i\beta;\:\alpha,\beta\in\mathbb{R},$$

\npgni where $1-\gamma=-\dfrac{2e\textrm{e}^{iv}}{(1-e\textrm{e}^{iv})}$ and $\dfrac{\nu}{2}(1-\gamma)=(1-e\textrm{e}^{iv}).$

\npgni This gives on our ellipse $\mathcal{E}_K$:

$$R+iS=-\frac{\mu}{\lambda}r+\lambda(1-e\textrm{e}^{iv})-2\lambda\ln 2-\lambda\ln(-e\textrm{e}^{iv}).$$

\npgni So on $\mathcal{E}_K$, since $r=a(1-e\cos v)$ here,

$$R=-\frac{\mu}{\lambda}r+\lambda(1-e\cos v)-2\lambda\ln2-\lambda\ln e,$$

\npgni $e$ being the eccentricity of $\mathcal{E}_K$. So on $\mathcal{E}_K$ to within additive constants,

$$R=-2\lambda\ln2-\lambda\ln e,\spc S=-\lambda(v+e\sin v),$$\vspace{-8mm}

\npgni $v$ being the eccentric anomaly. This establishes the constancy of the density function for the semi-classical state on the corresponding classical orbit.

\npgni We leave as an exercise in vector algebra the proof of equations $(\ast)$.

\npgni \textbf{Remark}

\npgni In fact, if $R$ achieves its global maximum on $\mathcal{E}_K$ and initial value of $R$ is $R(\vct x_0)$, where $\vct x_0=\mathbf{X}_{t=0}^0$, for all $\vct x_0\in\mathbb{R}^3$, the passage time from $\vct x_0$ to $\mathcal{E}_K$, $\tau(\vct x_0,\mathcal{E}_K)$, is the minimum value of the time $s$ such that

$$R_\textrm{max}-R(\vct x_0)=\int^s_0|\boldsymbol\nabla R|^2(\vct X_u^0)du.$$

\npgni This result is spoiled by the singularity $\Sigma$ of $(R+iS)$, $\Sigma=\{\nu\in\mathbb{C}:\mathcal{I}\nu=0,0<\mathcal{R}\nu<4\}$.

\begin{theorem}
\npgni $(R+iS)(\vct r)=-\dfrac{\mu}{\lambda}r+\lambda f(\nu)$, where

$$f(\nu)=\frac{\nu}{2}\left(1-\sqrt{1-\frac{4}{\nu}}\right)-\ln \nu-2\ln\left(1-\sqrt{1-\frac{4}{\nu}}\right).$$

\npgni So in the complex plane f has a cut $\{\nu\in\mathbb{C}: \mathcal{I}\nu=0,0<\mathcal{R}\nu<4\}=\Sigma$, across which $\sqrt{\zeta}$, $\zeta=1-\dfrac{4}{\nu}$ flips sign. In fact

$$f(\nu)=\frac{2}{(1+\sqrt\zeta)}-\ln\left(\frac{1-\sqrt{\zeta}}{1+\sqrt{\zeta}}\right)-\ln 4.$$

\npgni Setting $M(\Sigma)=\max_{\mathbf{r}\in\Sigma}R(\vct r)$ the above remark is valid for $\vct x_0$, if $R(\vct x_0)>M(\Sigma)$.
\end{theorem}

\begin{lemma}
$$M(\Sigma)=\lambda\left(2\left(\dfrac{1-e}{1+e}\right)-\ln4\right).$$
\end{lemma}

\begin{proof}

$$\Sigma=\{\nu=\nu_r+i\nu_i:\nu_i=0,0<\nu_r<4\}.$$

\npgni On $0<\nu_r<4$, $y=0$ and $r=\sqrt{x^2+z^2}$ giving

$$R=-\dfrac{\mu}{\lambda}r+\dfrac{\mu}{2\lambda}\left(r-\dfrac{x}{e}\right)-\lambda\ln{4}=R_{0}-\lambda\ln{4},$$

\npgni where $R_{0}=-\dfrac{\mu}{2\lambda}\left(r-\dfrac{x}{e}\right)$. $R_{0}$ is symmetric in $z$ so without loss of generality we can assume that $z>0$. Moreover $\dfrac{\partial R_{0}}{\partial z}<0$ for $z>0$ so the maximum of $R_{0}$ occurs when $z=0$ i.e. when $r=|x|$, $x<0$ and $0<\nu_r<4$ i.e at $x=-\dfrac{4\lambda^2e}{\mu(1+e)}$. Evaluating $R$ at this point gives the desired result.

\end{proof}\vspace{-8mm}

\begin{theorem}
\npgni The semi-classical mechanics of the astronomical elliptic state,

\npgni $\psi\sim\textrm{exp}\left(\frac{R+iS}{\hbar}\right)$ as $\hbar\sim0$, namely

$$\frac{d}{dt}\vct{X}_t^0=(\boldsymbol\nabla R+\boldsymbol\nabla S)(\vct{X}_t^0),\;\;t\ge0,$$\vspace{-5mm}

\npgni is well defined, if $\vct{x}_0=\vct{X}_{t=0}^0$ is such that

$$2\lambda\left(\dfrac{1-e}{1+e}\right)\le R(\vct{x}_0)+\lambda\ln4\le -\lambda\ln{e},$$\vspace{-5mm}

\npgni where, $\psi$ corresponds to eccentricity $e$, semi-major axis $a$ and $\lambda=\sqrt{\mu a}$, $\mu$ being the gravitational mass at the force centre O. Further

$$\ddot{\vct{X}_t^0}=-\boldsymbol\nabla V_{\textrm{eff}}(\vct{X}_t^0),$$\vspace{-12mm}

\npgni with\vspace{-5mm}

$$V_{\textrm{eff}}=V-|\boldsymbol\nabla R|^2,\;\;\;V=-\dfrac{\mu}{r}.$$

\npgni $|\boldsymbol\nabla R|^2$ is the first order quantum correction coming from the state $\psi$, so that the quantum curvature of the orbit reads

$$\kappa_{q}=\pm\frac{\vct{n}.\boldsymbol\nabla(V-|\boldsymbol\nabla R|^2)}{(|\boldsymbol\nabla R|^2+|\boldsymbol\nabla S|^2)},$$

\npgni $|\boldsymbol\nabla R|^2$ being zero on the classical orbit, $\vct{n}$ is the unit normal to the orbit at $\vct{X}_t^0$:

$$\vct{n}=\frac{{\vct{X}_t^0}^{''}}{|{\vct{X}_t^0}^{''}|}$$\vspace{-5mm}

\npgni $(')=\dfrac{d}{ds}=\dfrac{1}{|\vct{v}|}\dfrac{d}{dt}$, $\vct{v}=\boldsymbol\nabla R+\boldsymbol\nabla S$.

\end{theorem}

\npgni (See Ref. [28] for a generalisation).

\begin{proof}

Take the gradient of our energy equation, using  $\boldsymbol\nabla(2^{-1}\vct{a}^2)=(\vct{a}.\boldsymbol\nabla)\vct{a}+\vct{a}\times\textrm{curl}\vct{a}$, for any vector field $\vct{a}$.

\end{proof}\vspace{-8mm}

\npgni Inevitably our theory inherits some of the structure of Schrödinger quantum mechanics, especially as formulated by Ed Nelson in stochastic mechanics (see Ref. [18]). One of the most striking of these is that although our equations for $R$ and $S$ are non-linear nevertheless there is an underlying linear superposition principle for our solutions arising from the asymptotics. This is best illustrated by considering the superposition of two wave functions $\phi_{1}$ and $\phi_{2}$, for possibly different energies, satisfying our equations i.e. considering $\phi=\phi_{1}+\phi_{2}$, where

$$\phi_{\textrm{j}}\sim\textrm{exp}\left(\frac{R_{\textrm{j}}+iS_{\textrm{j}}}{\epsilon^2}\right),\;\;\textrm{j}=1,2\;\;\textrm{and}\;\; \phi\sim\textrm{exp}\left(\frac{R+iS}{\epsilon^2}\right)$$

\npgni as $\epsilon^2=\dfrac{\hbar}{m}\sim0$, $m$ being the WIMP mass.

\npgni Following Nelson it is easy to prove the exact formula,

$$2R=\epsilon^2\ln2+(R_{1}+R_{2})\ln\Bigg[\cosh\left(\dfrac{R_{1}-R_{2}}{\epsilon^2}\right)+\cos\left(\dfrac{S_{1}-S_{2}}{\epsilon^2}\right)\Bigg],$$

\npgni together with a similar formula for $S$. Letting $\epsilon\sim0$, we obtain

$$\phi\sim\dfrac{\rho_{1}\phi_{1}+\rho_{2}\phi_{2}}{\rho_{1}+\rho_{2}},\;\;\rho_{\textrm{j}}=\textrm{exp}\left(\dfrac{2R_{\textrm{j}}}{\epsilon^2}\right),\;\;\textrm{j}=1,2.$$

\npgni So, if $R_{\textrm{M}}=\textrm{max}(R_{1},R_{2})$, correct to first order in $\epsilon$,

$$R\sim R_{\textrm{M}}$$\vspace{-8mm}

\npgni and so $\phi=\left(c_{1}\chi_{_{R_{1}>R_{2}}}\phi_{1}+c_{2}\chi_{_{R_{2}>R_{1}}}\phi_{2}\right)$, with $c_{\textrm{j}}\ge0$, $c_{1}^2+c_{2}^2=1$, after normalisation. When $\phi_{\textrm{j}}$ corresponds to an astronomical elliptic state, $\textrm{j}=1,2$, by conservation of mass the $c_{\textrm{j}}^2$ are just the probabilities of WIMPs ending in the $\textrm{j}^{\textrm{th}}$ ring. Needless to say, leaning on Nelson (Ref. [19] - Quantum Fluctuations p. 90), it follows that when $R_{\textrm{M}}=R_{\textrm{j}}$, some j, for the above M,

$$\boldsymbol{\nabla}R=\boldsymbol{\nabla}R_{\textrm{M}},\;\;\;\boldsymbol{\nabla}S=\boldsymbol{\nabla}S_{\textrm{M}}\;\;\textrm{and}\;\;S=S_{\textrm{M}},\;\;\textrm{yielding our new solution},\;\;(R_{\textrm{M}}+iS_{\textrm{M}}).$$

\npgni So our superposition principle amounts to simply taking the maximum of the $R_{\textrm{j}}$'s and the corresponding $S_{\textrm{j}}$, as you might expect in the correspondence limit.

\subsection{Keplerian Elliptic Coordinates in 2-Dimensions}\vspace{-5mm}

\npgni The polar equation of an ellipse in the plane $z=0$, with eccentricity $\tilde e$ and semilatus rectum $\tilde\ell$, with focus at the origin $O$ and major axis parallel to the x-axis reads 

$$\frac{\tilde\ell}{r}=1+\tilde e\cos\theta,$$\vspace{-10mm}

\npgni which can be written as

$$\frac{\tilde\ell^2}{(1-\tilde e^2)}=x^2+\frac{2\tilde\ell\tilde e x}{(1-\tilde e^2)}+\frac{y^2}{(1-\tilde e^2)}$$\vspace{-10mm}

\npgni i.e.\vspace{-5mm}

$$(x+\tilde a\tilde e)^2+\frac{y^2}{(1-\tilde e^2)}=\tilde a^2=\frac{\tilde\ell^2}{(1-\tilde e^2)^2}.$$\vspace{-8mm}

\npgni So defining the coordinate $v$ by

$$x+\tilde a\tilde e=\tilde a\cos v,\spc y=\tilde a\sqrt{1-\tilde e^2}\sin v$$

\npgni and the coordinate $u$ by $u=\tilde e$ gives if we set $\tilde a=\dfrac{2ae}{(e+u)}$,

\eqn$$x=\frac{2ae}{(e+u)}(\cos v-u),\spc y=\frac{2ae}{(e+u)}\sqrt{1-u^2}\sin v.$$

\npgni A simple computation gives

$$r=\sqrt{x^2+y^2}=\frac{2ae}{(e+u)}(1-u\cos v).$$\vspace{-8mm}

\npgni So $v$ is the eccentric angle of a point $P$ on the ellipse with equation above, with $\tilde\ell=\dfrac{2ae}{(e+u)}(1-u^2),\;\tilde e=u.$

\npgni The following lemma in conjunction with Theorem 3.2  ensures the desired convergence to classical motion on the Keplerian ellipse.

\begin{lemma}

$$\dfrac{\partial R}{\partial x}=\dfrac{\partial R}{\partial z}=0\;\;\Rightarrow z=0.$$

\end {lemma}\vspace{3mm}

\begin{proof}
Using theorem 2.1 it is easy to establish that $\dfrac{\partial R}{\partial z}=0\Rightarrow z=0$ or\vspace{-5mm}

\npgni $\alpha=\mathcal{R}\sqrt{1-\dfrac{4}{\nu}}=-1$. If $\alpha=-1$, $\dfrac{\partial R}{\partial x}=0\Rightarrow\dfrac{1}{e}=0$, a contradiction, from which the\vspace{-3mm}

\npgni result follows.

\end{proof}\vspace {-5mm}

\npgni \textbf{N.B.} On Keplerian coordinates $(u,v)$, restricting $-e<u<1$ we have assumed $r<2a$, so the $u=\textrm{constant}$ contours are ellipses with eccentricity $u$, $0<u<1$, $v=\textrm{eccentric anomaly}$ on this ellipse. When $r>2a$ this changes because here $u<0$ and so the eccentricity of the ellipse in this region is $|u|$ and $v=\pi-\textrm{eccentric anomaly}$.

\npgni We can now recast our equations of motion, $\dot{\vct X}_t^0=(\boldsymbol\nabla R+\boldsymbol\nabla S)(\vct X_t^0),\vct X_t^0=(x,y)$, in terms of $(u,v)$:\vspace {-5mm}

$$\dot u=b_u(u,v),\spc\dot v=b_v(u,v).$$\vspace{-10mm}

\npgni We obtain for $c=\cos v$ and $s=\sin v$

$$b_u=\sqrt{\frac{\mu}{a^3}}\left(\frac{e+u}{2e}\right)^2\frac{\sqrt{1-u^2}\left(\sqrt{1-u^2}(e+c-s)-\sqrt{1-e^2}(u+c-s)\right)}{(1-uc)((e+u)c+eu+1)},$$

$$b_v=\sqrt{\frac{\mu}{a^3}}\frac{(e+u)}{4e^2}\frac{\left(e^2(u+c-s)+2e(1+uc)+u+c+s-\sqrt{1-u^2}\sqrt{1-e^2}(e+c+s)\right)}{(1-uc)((e+u)c+eu+1)}.$$

\npgni Observe that if we set $u=e$ in the equations for Keplerian coordinates we get the original Kepler ellipse for this problem. Observe that setting $u=e$ gives $b_u=0$ and

$$\dot v=\sqrt{\frac{\mu}{a^3}}\frac{1}{1-e\cos v}$$

\npgni i.e. Kepler’s equation for the eccentric anomaly,

$$v-e\sin v=\sqrt{\frac{\mu}{a^3}}t.$$

\npgni When $z=0$, the coordinates $u$ and $v$ enable us to simplify the expressions in $(x,y)$ for variables $\alpha$ and $\beta$, in 2-dimensions,

$$\alpha+i\beta=\sqrt{1-\frac{4}{\nu}},\spc \nu=\frac{\mu}{\lambda^2}\left(r-\frac{x}{e}-\frac{iy\sqrt{1-e^2}}{e}\right),\spc r=\sqrt{x^2+y^2},$$

$$\alpha=\left(\frac{1}{2}\sqrt{\frac{\left(er-x-\frac{4\lambda^2e}{\mu}\right)^2+(1-e^2)y^2}{(er-x)^2+(1-e^2)y^2}}+\frac{1}{2}\frac{\left(er-x-\frac{2\lambda^2e}{\mu}\right)^2+(1-e^2)y^2-\frac{4\lambda^4e^2}{\mu^2}}{(er-x)^2+(1-e^2)y^2}\right)\tph,$$

$$\beta=-\frac{2\lambda^2e\sqrt{1-e^2}y}{\alpha\mu((er-x)^2+(1-e^2)y^2)}:$$

$$\alpha=\frac{\sqrt{(1-u^2)(1-e^2)}}{(1+eu-(e+u)\cos v)},\;\;\;\;\;\;\;\beta=-\frac{(e+u)\sin v}{(1+eu-(e+u)\cos v)}.$$

\npgni Writing $\alpha=\sqrt{\frac{1}{2}(\sqrt{\alpha_1}+\alpha_2)}$, if we ask when is $\alpha_1$ is a perfect square, we discover $\alpha_1$ is a perfect square on the ellipses $\mathcal{E}_u$ with polar equation

$$r=\frac{2ae}{(e+u)}(1-u\cos v),\;\;\;-e<u<1,$$

\npgni with focus at the origin $O$, for fixed $u$, the eccentricity of $\mathcal{E}_u$, $0\le v\le2\pi$, the eccentric anomaly, given the above proviso when $u<0$. Hence our choice of $\tilde a=\dfrac{2ae}{(e+u)}$. 

\section{Levi-Civita, Kustaanheimo-Stieffel (LCKS) Transformations}
\subsection{LCKS Transformations in Classical Mechanics}

\npgni Consider the Kepler problem for $\zeta=\zeta(t)\in\mathbb{C}$, $t$ being the physical time,

$$\frac{d^2\zeta}{dt^2}=-\mu|\zeta|^{-3}\zeta,\spc t\ge0.$$

\npgni Define $w\in\mathbb{C}$ by $\zeta=w^2$ and a new time variable $s$ by $\dfrac{ds}{dt}=|\zeta(t)|^{-1}$. Then a simple calculation yields that the equation for energy conservation

$$2^{-1}\left|\frac{d\zeta}{dt}\right|^2-\frac{\mu}{|\zeta|}=E,\spc E<0$$

\npgni a constant here, reduces to\vspace{-3mm}

$$2^{-1}\left|\frac{dw}{ds}\right|^2-\frac{E}{4}|w|^2=\frac{\mu}{4}$$\vspace{-5mm}

\npgni and working slightly harder for a constant $\omega$

$$\frac{d^2w(s)}{ds^2}+\omega^2w(s)=0.$$\vspace{-5mm}

\npgni This is the original result due to Levi-Civita; energy conservation and equation of motion for an isotropic oscillator in 2-dimensions with energy $E'=\dfrac{\mu}{4}$ and circular frequency $\omega=\sqrt{-\dfrac{E}{4}}$.
Here we investigate the transformation $(x+iy)\to(x+iy)^2$ or for $(x+iy)=re^{i\theta}$, $r\to r^2$, $\theta\to2\theta$ and the time-change, $\dfrac{ds}{dt}=\dfrac{1}{r}=\sqrt{\dfrac{a}{\mu}}\dfrac{dv}{dt}$, $v$ being the eccentric anomaly on $\mathcal{E}_K$, the ficticous time $s$, and physical time $t$ (Ref. [25]). 

\begin{lemma}
\npgni Under the LCKS Transformation the Kepler ellipse with focus at the 

\noindent origin $O$\vspace{-3mm}

$$\frac{\ell}{r}=1+e\cos\theta,\spc0<e<1,\spc\ell>0,$$\vspace{-8mm}

\npgni i.e.\;$r=a(1-e\cos v),$ is transformed to the ellipse with origin at the centre of the ellipse

$$\frac{x^2}{p^2}+\frac{y^2}{q^2}=1,\spc p^2=\frac{\ell}{(1+e)},\spc q^2=\frac{\ell}{(1-e)}.$$

$$(x,y)=a(\cos v -e, \sqrt{1-e^2}\sin v)\to\left(\sqrt{a(1-e)}\cos\left(\frac{v}{2}\right),\sqrt{a(1+e)}\sin\left(\frac{v}{2}\right)\right).$$
\end{lemma}\vspace{-5mm}

\npgni N.B. The major axis for the transformed ellipse is parallel to the y-axis.

\npgni There is no LCKS transformation in 3-dimensions; there is one in 4-dimensions using quaternions, but here we work in 2-dimensions. Here $r=\sqrt{x^2+y^2}$ and

\noindent $\nu=\dfrac{\mu}{\lambda^2}\left(r-\dfrac{x}{e}-\dfrac{iy\sqrt{1-e^2}}{e}\right)$, for our 2-dimensional Kepler-Coulomb problem. We presuppose that conservation of angular momentum fixes the plane of motion to be, $z=0$, with cartesian coordinates $(x,y)$. Our main result for our semi-classical analysis for the astronomical elliptic state, $\psi_{\textrm{sc}}$, for the Kepler ellipse, $\mathcal{E}_K$, with polar equation\vspace{-3mm}

\npgni has eccentricity $e$, $0<e<1$, with semi-major axis, $a=\dfrac{\lambda^2}{\mu}$, semi-latus rectum, $\ell=a(1-e^2)$, in the plane $z=0$, with semi-major axis parallel to the x-axis, with polar equation:

$$\frac{\ell}{r}=1+e\cos\theta,$$\vspace{-5mm}

\npgni $\ell=a(1-e^2)$, $e=\sqrt{1+\dfrac{2h^2E}{\mu^2}}$, $0<e<1$, $a$ being the semi-major axis parallel to the $x-\textrm{axis}$, $e$ the orbital eccentricity, $h$, $E$ and $\mu$ as usual being angular momentum, energy of WIMPish particles per unit mass and $\mu$ the gravitational mass at the force centre of $\mathcal{E}_K$, $(ae,0)$, relative to the centre of the elliptical origin O.\vspace{3mm}

\begin{theorem}
For our Keplerian coordinates $(u,v)$ defined above, converting to polar coordinates $(r,\theta)$: $x=r\cos{\theta}$, $y=r\sin{\theta}$, reveals,

$$u=u(r,\theta),\;\;\;\;\;\;\;\tan{\dfrac{v}{2}}=\sqrt{\dfrac{1-u}{1+u}}\tan{\dfrac{\theta}{2}}.$$\vspace{-5mm}

\npgni $u(r,\theta)$ being the solution to a simple quadratic equation.\vspace{-3mm}

\npgni Moreover, as defined above, for\vspace{-3mm}

$$\nu=\frac{\mu}{\lambda^2}\left(r-\frac{x}{e}-\frac{iy\sqrt{1-e^2}}{e}\right),\spc r=\sqrt{x^2+y^2},\spc \alpha+i\beta=\sqrt{1-\frac{4}{\nu}},$$

$$|\boldsymbol{\nabla}R|^2=\dfrac{\mu^2(1+e\cos{v})(1-eu-\sqrt{1-u^2}\sqrt{1-e^2})}{2\lambda^2e^2(1-u\cos{v})}.$$

\npgni So $|\boldsymbol{\nabla}R(x,y)|=0$ iff
$(x,y)\in\mathcal{E}_K$ and as expected $R(\vct{X}_t^0)\nearrow R_{\textrm{max}}$ as $t\nearrow\infty$ guarantees convergence to the classical Keplerian ellipse as long as $R(\vct x_0)>M(\Sigma)$. The formulae for $b_{u}$ and $b_{v}$ being as given above.

\end{theorem}

\begin{proof}
We warn the reader that the simplifications in this proof are heavily dependent on Levi-Civita's transformations. Using this transformation together with the definition of Keplerian coordinates $(u,v)$ we can deduce that $u$ satisfies the quadratic equation,

$$(r\cos{\theta}+2ae)u^2+(er\cos{\theta}+r)u+(r-2a)e=0,$$\vspace{-7mm}

\npgni where $a=\dfrac{\lambda^2}{\mu}$. Solving this equation gives $u(r,\theta)$ and the formula for $v$ follows from the definition of $(u,v)$. In 2-dimensions $\boldsymbol{\nabla}R$ is given by

$$\boldsymbol{\nabla}R=-\dfrac{\mu}{2\lambda}\left((1+\alpha)\dfrac{x}{r}+\dfrac{(1-\alpha)}{e},(1+\alpha)\dfrac{y}{r}+\dfrac{\beta}{e}\sqrt{1-e^2}\right).$$

\npgni The expression for $|\boldsymbol{\nabla}R|^2$ follows after straight-forward but lengthy algebra. Observing that $|\boldsymbol{\nabla}R(x,y)|=0$ iff $u=e$ completes the proof.

\end{proof}

\subsection{Quantum Harmonic Oscillators}\vspace{-5mm}

\npgni To continue we need to note that the elliptic harmonic oscillator state in 2-dimensions is

$$\psi_n^{\textrm{HO}}=\textrm{exp}\left(-\frac{n\omega r^2}{2\tilde\lambda}\right)H_n(\sqrt n\tilde u),\spc n=1,2\dots,$$

\npgni where $\tilde u=\sqrt{\dfrac{\omega}{\tilde\lambda}\left((1-\alpha)\dfrac{x^2}{2}+(1+\alpha)\dfrac{y^2}{2}-i\beta xy\right)}$, $\alpha^2-\beta^2=1$, for potential 

\npgni $V(r)=\dfrac{1}{2}\omega^2r^2$, $H_n$ a Hermite polynomial (Ref. [12]).

\npgni For the quantum Hamiltonian

$$H=2^{-1}\vct P^2+V(\vct Q),$$

$$H\psi_n^{\textrm{HO}}=\omega\hbar(n+1)\psi_n^{\textrm{HO}},\spc\tilde\lambda=n\hbar=n\epsilon^2,$$

\npgni $\tilde\lambda$ being fixed as we take the limit $\epsilon\searrow0$, $n\nearrow\infty$, so the limiting energy here is $\tilde\lambda\omega$. We need the result that for fixed $\tilde\lambda$,

$$\lim_{n\nearrow\infty}\frac{H_n'(\sqrt{n}u)}{\sqrt{n}H_n(\sqrt{n}u)}=u-\sqrt{u^2-2},$$\vspace{-5mm}

\npgni giving $\psi_n^{\textrm{HO}}\cong\textrm{exp}\left(\dfrac{\tilde R+i\tilde S}{\epsilon^2}\right)$ as $n\sim\infty$, where

$$(\tilde R+i\tilde S)(x,y)=-\frac{1}{2}\omega r^2+\frac{\tilde\lambda}{2}\tilde u^2\left(1-\sqrt{1-\frac{2}{\tilde u^2}}\right)+\tilde\lambda\ln\left(\tilde u+\tilde u\sqrt{1-\frac{2}{\tilde u^2}}\right).$$

\begin{theorem}
Under the LCKS transformation:

$$\nu(x,y)\to\dfrac{\mu}{\lambda^2}\left(x^2+y^2-\frac{1}{e}(x^2-y^2)-\dfrac{2ixy\sqrt{1-e^2}}{e}\right).$$\vspace{-5mm}

\npgni So defining $\omega$ by $\dfrac{\mu}{\lambda^2}=\dfrac{\omega}{\tilde\lambda}$, if $\dfrac{C\omega}{2}=\dfrac{\mu}{\lambda}$ and $C\tilde\lambda=2\lambda$, a simple calculation yields under LCKS\vspace{-5mm}

$$(R+iS)(x, y)\to C(\tilde R+i \tilde S)(x, y),$$\vspace{-8mm}

\npgni $C$ is a change of time scale here. Needless to say in 2-dimensions

$$\boldsymbol\nabla\tilde R.\boldsymbol\nabla\tilde S=0,\spc 2^{-1}\left(|\boldsymbol\nabla\tilde S|^2-|\boldsymbol\nabla\tilde R|^2\right)+V=\tilde\lambda\omega,$$\vspace{-5mm}

\npgni where $V=\dfrac{1}{2}\omega^2r^2$.
\end{theorem}

\npgni The Levi-Civita-Kustaanheimo-Stieffel transformation in 2-dimensions gives a hint of the following result in 3-dimensions.

\begin{theorem}
The isotropic harmonic oscillator elliptic state is, up to normalisation,

$$\psi_n^{\textnormal{HO}}=\textrm{exp}\left(-\frac{\omega r^2}{2\tilde\lambda}\right)H_n(\sqrt{n}\tilde u),$$\vspace{-3mm}

\npgni where $r^2=x^2+y^2+z^2$, $\tilde u=\sqrt{\dfrac{\omega}{\tilde\lambda}\left((1-\alpha)\dfrac{x^2}{2}+(1+\alpha)\dfrac{y^2}{2}-i\beta xy\right)}$, $\alpha=\dfrac{1}{e}$,\vspace{-5mm} 

\npgni $\beta=\dfrac{\sqrt{1-e^2}}{e}$, $\alpha^2-\beta^2=1$, $e$ being the eccentricity of the ellipse with equations

$$z=0,\spc\frac{x^2}{p^2}+\frac{y^2}{q^2}=1,$$\vspace{-8mm}

\npgni as above. $H_n$ being a Hermite polynomial. 

\npgni For quantum Hamiltonian $H=2^{-1}\vct P^2+2^{-1}\omega^2\vct Q^2$,

$$H\psi_n^{\textnormal{HO}}=\omega\left(\tilde\lambda+\frac{3}{2}\hbar\right)\psi_n^{\textnormal{HO}},\spc\tilde\lambda=n\hbar=n\epsilon^2.$$
\end{theorem}

\begin{proof}
First $\psi_1(z)=\textrm{exp}\left(-\dfrac{\omega z^2}{2\tilde\lambda}\right)$ is the ground state of the harmonic oscillator in 1-dimension and if

$$\psi_2(x,y)=\textrm{exp}\left(-\frac{\omega(x^2+y^2)}{2\tilde\lambda}\right)H_n(\sqrt{n}\tilde u),$$

$$\boldsymbol\nabla\psi_1.\boldsymbol\nabla\psi_2=0, \;\;\textrm{so}\;\; \Delta(\psi_1\psi_2)=\psi_1(\Delta\psi_2)+(\Delta\psi_1)\psi_2.$$

\npgni Further $\Delta\tilde u=0$, so $\Delta H_n(\sqrt{n}\tilde u)=nH_n''(\sqrt{n}\tilde u)|\boldsymbol\nabla\tilde u|^2$. Also,

\begin{flushleft}
\;\;$\Delta\psi_2 = \Delta\left(\textrm{exp}\left(-\dfrac{\omega(x^2+y^2)}{2\tilde\lambda}\right)H_n(\sqrt{n}\tilde u)\right)$
\end{flushleft}

\begin{flushleft}
$\spc\;\;=\textrm{exp}\left(-\frac{\omega(x^2+y^2)}{2\tilde\lambda}\right)\Delta H_n(\sqrt{n}\tilde u)+H_n(\sqrt{n}\tilde u)\Delta\textrm{exp}\left(-\frac{\omega(x^2+y^2)}{2\tilde\lambda}\right)$
\end{flushleft}

\begin{flushleft}
$\spc\spc\spc\spc\spc\spc\spc+2\boldsymbol\nabla\left(\textrm{exp}\left(-\frac{\omega(x^2+y^2)}{2\tilde\lambda}\right)\right).\boldsymbol\nabla H_n(\sqrt{n}\tilde u).$
\end{flushleft}

\npgni A computation using the properties of Hermite polynomials yields the result

$$-\frac{\hbar^2}{2}\Delta\psi_n^{\textrm{HO}}+\frac{\omega^2}{2}(x^2+y^2+z^2)\psi_n^{\textrm{HO}}=\omega\left(\tilde\lambda+\frac{3}{2}\hbar\right)\psi_n^{\textrm{HO}},\spc\tilde\lambda=n\hbar.$$
\end{proof}\vspace{-5mm}

\npgni As an application we spell out the connection with the restricted 3-body problem for the Trojan asteroids where we uncover two new constants of motion/integration for the usual linearised problem. It is implicit in the next section that the quantum Hamiltonian, $\hat{H}$, in the rotating frame for an asteroid in the Sun/Jupiter 2-body problem is the Jacobi integral after quantisation. This follows from the fact that, if $U$ denotes the unitary transformation to the rotating frame, then

$$U^{\dag}\left(i\hbar\frac{\partial}{\partial t}\right)U=i\hbar\frac{\partial}{\partial t}+\omega\hat{L_{3}},\;\;\;\;\;(\textrm{HJI}),$$

\npgni where $\omega$ is the angular velocity of the axes and  $\hat{L_{3}}$ is the angular momentum operator normal to the plane of orbits. The final term which appears in $\hat{H}$ couples the cartesian coordinates of the asteroid, as you will see. c.f. Whittaker (Ref. [29]).

\subsection{The Trojan Asteroids}\vspace{-5mm}

\npgni As is well known, for the Sun(S)/Jupiter(J) system, Lagrange's equlateral triangle solution of the 3-body problem yields two points of stable equilibrium, $\mathscr{L}_{4,5}$, for a Trojan asteroid (A), $60^{\circ}$ ahead or behind J and co-orbital with it. Here the Sun and Jupiter are rotating about their mass centre, $O$, with angular speed $\omega$, SAJ forming an equilateral triangle with side length $a_{0}$. The equations of motion can be deduced from the Hamiltonian in the rotating frame, usually referred to as $K$, the Jacobi integral, which for $\overset{\rightharpoonup}{OA}=(\tilde x,\tilde y)$ read:-

$$\ddot{\tilde x}-2\omega\dot{\tilde y}-\omega^{2}\tilde x=-\frac{\mu_{1}(\tilde x+r_{1})}{((\tilde x+r_{1})^{2}+{\tilde y}^{2})^{3/2}}-\frac{\mu_{2}(\tilde x-r_{2})}{((\tilde x-r_{2})^{2}+{\tilde y}^{2})^{3/2}},$$

$$\ddot{\tilde y}+2\omega\dot{\tilde x}-\omega^{2}\tilde y=-\frac{\mu_{1}\tilde y}{((\tilde x+r_{1})^{2}+{\tilde y}^{2})^{3/2}}-\frac{\mu_{2}\tilde y}{((\tilde x-r_{2})^{2}+{\tilde y}^{2})^{3/2}},$$

\npgni $\overset{\rightharpoonup}{O\tilde x} \parallel \overset{\rightharpoonup}{OJ}$, where $\mu_{1}$ is the gravitational mass of the Sun, $\mu_{2}$ the gravitational mass of Jupiter, $r_{1}$ the distance of the Sun to $O$ and $r_{2}$ the distance of Jupiter to $O$. $(\tilde x,\tilde y)$ is related to the fixed frame coordinates of $\overset{\rightharpoonup}{OA}$ by
 
$$x=\tilde x \cos(\omega t)-\tilde y \sin(\omega t)\spc,\spc y=\tilde x \sin(\omega t)+\tilde y \cos(\omega t)$$

\npgni and ${\omega}^{2}=\dfrac{(\mu_{1}+\mu_{2})}{(r_{1}+r_{2})^{3}}$,  $r_{1}+r_{2}=a_{0}$. These equations admit a fixed point solution,

$$(\tilde x,\tilde y)=(c,d)=\left(\frac{r_{2}-r_{1}}{2},\pm \frac{\sqrt{3}}{2}(r_{1}+r_{2})\right),$$

\npgni for the asteroid equilibrium points $\mathscr{L}_{4,5}$.

\npgni An investigation of the stability near these points for $\tilde x=c+\delta(t)$ and $\tilde y = d+\epsilon(t)$ leads to the equations:

$$\ddot{\delta}-2\omega \dot{\epsilon}-\frac{3}{4}{\omega}^2\delta-{\Omega}^2\epsilon=0,$$

$$\ddot{\epsilon}+2\omega \dot{\delta}-\frac{9}{4}{\omega}^2\epsilon-{\Omega}^2\delta=0,$$

\npgni where ${\Omega}^2=\dfrac{3\sqrt{3}(\mu_{1}-\mu_{2})}{4(r_{1}+r_{2})^3}$, given that $\mu_{1}\ge\mu_{2}$.

\npgni These equations admit a constant of the motion, related to the Jacobi constant,

$$\tilde C=\frac{1}{2}({\dot{\delta}^2}+{\dot{\epsilon}^2})-\frac{3}{8}{\omega}^2({\delta}^2+3{\epsilon}^2)-\Omega^2\epsilon\delta.$$

\npgni This suggests that we look at a rotation of coordinates: $(\delta,\epsilon)\rightarrow(X,Y)$ via

$$\delta=X\cos \gamma - Y\sin \gamma \spc;\spc \epsilon=X\sin \gamma + Y\cos \gamma.$$

\npgni Using the above constant and requiring the $XY$ term to vanish leads to the condition,

$$\tan 2\gamma=-\dfrac{\sqrt{3}(\mu_{1}-\mu_{2})}{\mu_{1}+\mu_{2}},$$

\npgni and a constant of the motion in the $(X,Y)$ coordinates,

$$C_{0}=\frac{1}{2}({\dot{X}^2}+{\dot{Y}^2})-\omega_{X}^2X^2-\omega_{Y}^2Y^2,$$

\npgni where

$$\omega_{X}^2=\left(\dfrac{3}{4}-\left(\dfrac{3}{8}+\dfrac{2\Omega^4}{3\omega^4}\right)\cos 2\gamma\right)\omega^2,$$

$$\omega_{Y}^2=\left(\dfrac{3}{4}+\left(\dfrac{3}{8}+\dfrac{2\Omega^4}{3\omega^4}\right)\cos 2\gamma\right)\omega^2.$$

\npgni We note here that $\omega_{X}^2+\omega_{Y}^2=\dfrac{3}{2}\omega^2$ and $\omega_{X}^2\omega_{Y}^2=\dfrac{27\mu_{1}\mu_{2}}{16(\mu_{1}+\mu_{2})^2}\omega^4$.

\npgni Moreover $X$ and $Y$ satisfy

$$\ddot{X}-2\omega\dot{Y}-2\omega_{X}^2X=0\spc;\spc\ddot{Y}+2\omega \dot{X}-2\omega_{Y}^2Y=0.$$

\npgni Assuming solutions of the form $X=A\textrm{e}^{\lambda t}$ and $Y=B\textrm{e}^{\lambda t}$ leads to

$$\lambda^4+\omega^2\lambda^2+4\omega_{X}^2\omega_{Y}^2=0,$$

\npgni with the solutions taking the form $X=2\omega\lambda C\textrm{e}^{\lambda t}$,  $Y=(\lambda^2-2\omega_{X}^2)\textrm{e}^{\lambda t}$, for a constant $C$.

\npgni Given that there are four distinct roots for $\lambda$ the general solution for $X$ and $Y$ is

$$X=2\omega {\sum_{\textrm{i}=1}^{4}}\lambda_\textrm{i}C_\textrm{i}\textrm{e}^{\lambda_\textrm{i}t}\spc;\spc Y={\sum_{\textrm{i}=1}^{4}}(\lambda_\textrm{i}^2-2\omega_{X}^2)C_\textrm{i}\textrm{e}^{\lambda_\textrm{i}t}$$

\npgni If we suppose further that all the roots are imaginary i.e. $\lambda_{1}=\alpha i$, $\lambda_{2}=-\alpha i$, $\lambda_{3}=\beta i$ and $\lambda_{4}=-\beta i$ with $\alpha\ne\beta$ and appropriate $C_{i}$ there exist elliptical orbit solutions

$$X=4\omega\alpha C\sin \alpha t\spc;\spc Y=2(\alpha^2+2\omega_{X}^2)C\cos \alpha t$$

\npgni and the corresponding solution with $\beta$. We note here that, $\alpha^4-\omega^2\alpha^2+4\omega_{X}^2\omega_{Y}^2=0$ and the same for $\beta$. It is also worth noting that

$$\alpha=\dfrac{1}{\sqrt{2}}\left(1+\sqrt{1-\dfrac{27\mu_{1}\mu_{2}}{(\mu_{1}+\mu_{2})^2}}\right)^{\frac{1}{2}}\omega\;\;;\;\;\beta=\dfrac{1}{\sqrt{2}}\left(1-\sqrt{1-\dfrac{27\mu_{1}\mu_{2}}{(\mu_{1}+\mu_{2})^2}}\right)^{\frac{1}{2}}\omega,$$

\npgni provided that $\dfrac{\mu_{1}\mu_{2}}{(\mu_{1}+\mu_{2})^2}<\dfrac{1}{27}$, i.e. all solutions are oscillatory. This condition is\vspace{-3mm}

\npgni satisfied for the Sun/Jupiter system.

\npgni For now we consider just this elliptical orbit which requires special initial conditions which we come to later. Evidently, by inspection,

$$\ddot{X}=-\alpha^2X\spc\textrm{and}\spc\ddot{Y}=-\alpha^2Y$$\vspace{-8mm}

\npgni for the above root $\alpha>\beta$. So we have a 2-dimensional isotropic oscillator in this case with constants of motion,

$$2^{-1}{\dot{X}}^2+2^{-1}\alpha^2X^2=E_{X}\spc,\spc2^{-1}{\dot{Y}}^2+2^{-1}\alpha^2Y^2=E_{Y}$$

\npgni and an elliptical orbit with equation $\dfrac{X^2}{A^2}+\dfrac{Y^2}{B^2}=1$, where $A^2=4\omega^2\alpha^2C^2$,\vspace{-3mm}

\npgni $B^2=(\alpha^2-2\omega_{X}^2)^2C^2$, $C$ having to be small e.g. $C/r_{2}=O(\mu_{2}/(\mu_{1}+\mu_{2}))$.

\npgni The eccentricity of this ellipse, $e$, is given by,

$$e=\sqrt{1-\left(\dfrac{\textrm{min}(|A|,|B|)}{\textrm{max}(|A|,|B|)}\right)^2}=\sqrt{1-\dfrac{A^2}{B^2}},$$

\npgni in our case. So knowing $C$ (the strength of perturbation), the above theorem gives $\mathbf{\nabla} R$ and $\mathbf{\nabla} S$ explicitly for this linearised problem guaranteeing that the semi-classical spiral orbit converges, in this case, to the above elliptical orbit as it should. Of course, we could choose $\beta$ instead of $\alpha$ but we prefer just now to concentrate on $\alpha$, the near resonance solution for the Sun/Jupiter system,

$$\alpha=0.996758\,\omega\spc,\spc\beta=0.080464\,\omega.$$\vspace{-8mm}

\npgni If we linearise the $Z$ problem, we also obtain for small $Z$,

$$\ddot{Z}=-\omega ^2Z$$\vspace{-8mm}

\npgni so for the above $C$ and choice of $\alpha$, correct to first order in $\mu_{2}/(\mu_{1}+\mu_{2})$, formally,

$$\ddot{X}\cong-\omega ^2X\spc,\spc\ddot{Y}\cong-\omega ^2Y\spc,\spc\ddot{Z}\cong-\omega ^2Z$$

\npgni i.e. a 3-dimensional isotropic oscillator for which again we know $R$ and $S$ explicitly. For the solution correct to first order, needless to say we have the full SU(3) symmetry with 8 constants of the motion corresponding to the SU(3) generators, 3 energies, 3 components of angular momentum and 2 others (Ref. [24]). This tempted us to look for further constants of the motion for the linearised Trojan asteroid problem, in general.

\begin{theorem}(Hidden Constants)

\npgni For the system

$$\ddot{X}-2\omega \dot{Y}-2\omega_{X}^2X=0\spc;\spc\ddot{Y}+2\omega \dot{X}-2\omega_{Y}^2Y=0,$$

\npgni where $\omega_{X}^2+\omega_{Y}^2=\dfrac{3}{2}\omega^2$ and $\omega_{X}^2\omega_{Y}^2=\dfrac{27\mu_{1}\mu_{2}}{16(\mu_{1}+\mu_{2})^2}\omega^4$, in addition to the Jacobi constant,\vspace{3mm} $\frac{1}{2}({\dot{X}^2}+{\dot{Y}^2})-\omega_{X}^2X^2-\omega_{Y}^2Y^2=C_{0}$, there are two further constants of the motion:-

$$\omega_{Y}^2\{(\alpha^2+2\omega_{X}^2)\dot{X}-2\omega\alpha^2Y\}^2+\alpha^2\omega_{X}^2\{2\omega\dot{Y}+(\alpha^2+2\omega_{X}^2)X\}^2=D_{0}$$

\npgni and similarly for $\beta$; $\alpha$, $\beta$ being the real roots of $\alpha^4-\omega^2\alpha^2+4\omega_{X}^2\omega_{Y}^2=0$, provided that\vspace{3mm} $\dfrac{\mu_{1}\mu_{2}}{(\mu_{1}+\mu_{2})^2}<\dfrac{1}{27}$.

\end{theorem}\vspace{5mm}

\begin{proof}
(Outline). For the system defined above the solutions for $X$ and $Y$ together with their time derivatives can be expressed in matrix form:

$$\begin{pmatrix}{X}\\{\dot X}\\{Y}\\{\dot Y}\end{pmatrix}=\begin{pmatrix}2\omega\lambda_{1} & 2\omega\lambda_{2} & 2\omega\lambda_{3} & 2\omega\lambda_{4}\\2\omega\lambda_{1}^2 & 2\omega\lambda_{2}^2 & 2\omega\lambda_{3}^2 & 2\omega\lambda_{4}^2\\\lambda_{1}^2-2\omega_{X}^2 & \lambda_{2}^2-2\omega_{X}^2 & \lambda_{3}^2-2\omega_{X}^2 & \lambda_{4}^2-2\omega_{X}^2\\ \lambda_{1}(\lambda_{1}^2-2\omega_{X}^2) & \lambda_{2}(\lambda_{2}^2-2\omega_{X}^2) & \lambda_{3}(\lambda_{3}^2-2\omega_{X}^2) & \lambda_{4}(\lambda_{4}^2-2\omega_{X}^2) \end{pmatrix}\begin{pmatrix}{C_{1}\textrm{e}^{\lambda_{1}t}}\\{C_{2}\textrm{e}^{\lambda_{2}t}}\\{C_{3}\textrm{e}^{\lambda_{3}t}}\\{C_{4}\textrm{e}^{\lambda_{4}t}}\end{pmatrix}.$$

\npgni Inverting this matrix equation gives expressions for $\textrm{e}^{\lambda_{\textrm{i}}t},\;\textrm{i}=1,2,3,4$ in terms of $X,\dot{X},Y$\vspace{3mm}
\npgni and $\dot{Y}$. When $\dfrac{\mu_{1}\mu_{2}}{(\mu_{1}+\mu_{2})^2}<\dfrac{1}{27}$, $\lambda_{1}=\alpha i$, $\lambda_{2}=-\alpha i$, $\lambda_{3}=\beta i$ and $\lambda_{4}=-\beta i$, where $\alpha$\vspace{-3mm}
\npgni and $\beta$ are the real roots of $\alpha^4-\omega^2\alpha^2+4\omega_{X}^2\omega_{Y}^2=0$. Assuming that $\omega_{X}(\alpha^2-\beta^2)\ne0$,\vspace {-3mm}
\npgni constants of the motion follow from the fact that $\textrm{e}^{i\alpha t}\textrm{e}^{-i\alpha t}=1$ and similarly for $\beta$.

\end{proof}\vspace {-5mm}

\npgni Further details are given in (Ref. [28]) which also includes a generalisation of Lagranges' equilateral triangle solution of the 3-body problem and the semi-classical analysis showing the convergence of particles to the region centred on the Lagrange points $\mathscr{L}_{4,5}$. It is inconceivable that asteroids on the above orbit will not interact with the Hildian asteroids in $\frac{3}{2}$ resonance with Jupiter, particularly if the orbital eccentricity of the Hildian is $\left(\frac{2}{3}\right)^{-\frac{2}{3}}-1$ and it is at aphelion near Jupiter's orbit.

\npgni The condition for $\alpha$-isotropy i.e. isotropy in 2-dimensions with circular frequency $\alpha$, always assuming that $\omega_{X}(\alpha^2-\beta^2)\ne0$ reduces to:-

$$\beta(\alpha^2+2\omega_{X}^2)\dot{X}_{0}-2\omega\alpha^2\beta Y_{0}+2\omega_{X}^2(2\omega\dot{Y}_{0}+(\alpha^2+2\omega_{X}^2)X_{0})=0,$$\vspace{-8mm}

\npgni with a similar result for $\beta$-isotropy, obtained by replacing $\alpha$ by $\beta$. The constants of motion, for energy and angular momentum for $\alpha$-isotropy, are\vspace {-5mm}

$$E=E_{X}+E_{Y}=\dfrac{\alpha^2}{2}(A^2+B^2)=Q_{1},\;\;L^2=\alpha^2A^2B^2,\;\;E_{X}-E_{Y}=\dfrac{\alpha^2}{2}(B^2-A^2)=Q_{0},$$\vspace {-10mm}

\npgni and $Q_{XY}=\alpha^2XY+p_{X}p_{Y}=0$, with a similar result for the remaining $Q's$.\vspace{-3mm}

\npgni Needless to say the above analysis has applications to other physical problems, e.g. Foucault's pendulum.\vspace{-3mm}

\section{Newton's Revolving Orbits}\vspace {-5mm}

\npgni Building on the results of the previous sections we now discuss the semi-classical system defined by a given central potential, $V_0$, with the addition of an inverse cube radial force. Specifically we shall consider the systems defined by $V_K=V_0+\dfrac{(1-k^{2})L^2}{2r^2} $, $k\in\mathbb{R}$, for $V_0=-\dfrac{\mu}{r}$ and $V_0=\dfrac{1}{2}\omega^2r^2$, $r^2=x^2+y^2$. $L$ is the angular momentum for the system defined by $V_0$. This system was well known to Newton (Ref. [3]) and discussed in detail by Lynden-Bell (Ref. [13]). In essence the result derived for the classical motion is; if $r(\theta)$ describes the orbit under $V_0$ then $r\left(\dfrac{\theta}{k}\right)$ defines the orbit under $V_K$. The key to our semi-classical analysis will be the use of the Hamiltonian and a complex constant of the motion (which seems to have been given little attention in the literature) coupled with a simple transformation in the plane. We shall then observe that the semi-classical motion converges beautifully to the classical motion in the elliptical case as described above.\vspace{-2mm}

\subsection{Constants of the Motion}

Consider a particle of unit mass moving in the $(X,Y)$ plane and governed by the Hamiltonian\vspace{-5mm}

$$H=2^{-1}\vct P^2+V_0(\vct Q).$$\vspace{-8mm}

\npgni Using polar coordinates $(\rho, \phi)$ and defining $\vct P=(P_X,P_Y)=(\dot X,\dot Y)$ it can be shown that the systems above have constants of the motion, in time, defined respectively by:-

$$V_0=-\frac{\mu}{\rho}:\spc\frac{1}{2}(P_X^2+P_Y^2)-\frac{\mu}{\rho}=E,\spc E<0,$$

$$\frac{\cos\left(\frac{\phi}{2}\right)P_X+\sin\left(\frac{\phi}{2}\right)P_Y-i\sqrt{-2E}\cos\left(\frac{\phi}{2}\right)}{\cos\left(\frac{\phi}{2}\right)P_Y-\sin\left(\frac{\phi}{2}\right)P_X-i\sqrt{-2E}\sin\left(\frac{\phi}{2}\right)}=- i\frac{\sqrt{1-e}}{\sqrt{1+e}}.$$\vspace{-5mm}

\npgni $E=-\dfrac{\mu^2}{2\lambda^2}$, is the energy and $0<e<1$ is the eccentricity.

$$V_0=\frac{1}{2}\omega^2\rho^2:\spc\frac{1}{2}(P_X^2+P_Y^2)+\frac{1}{2}\omega^2\rho^2=E,\spc E>0,$$

$$\frac{P_X-i\omega X}{P_Y-i\omega Y}=-\frac{a}{b}i,\spc a,b\in\mathbb{R}^+.$$

\npgni $E=\frac{1}{2}(a^2+b^2)\omega^2$, is the energy for angular frequency $\omega$.

\npgni The two systems are connected by the LCKS transformation and corresponding time-change in 2-dimensions as described in section 3.

\npgni We now move to the $(x,y)$ plane and polar coordinates $(r,\theta)$ via the transformation $X=r\cos\left(\frac{\theta}{k}\right)$ and $Y=r\sin\left(\frac{\theta}{k}\right)$. If we identify $\vct p=(p_x,p_y)=(\dot x,\dot y)$ it is not difficult to show that

$$P_X=\frac{1}{r}(xp_x+yp_y)\cos\left(\frac{\theta}{k}\right)-\frac{L}{r}\sin\left(\frac{\theta}{k}\right),$$

$$P_Y=\frac{1}{r}(xp_x+yp_y)\sin\left(\frac{\theta}{k}\right)+\frac{L}{r}\cos\left(\frac{\theta}{k}\right),$$\vspace{-5mm}

\npgni where $L$ is the angular momentum of the underlying system defined by $V_0(\rho)$. We note that for $V_0=-\dfrac{\mu}{\rho}$,
$L=\lambda\sqrt{1-e^2}$ and for $V_0=\frac{1}{2}\omega^2\rho^2$, $L=ab\omega$. Using this transformation, the above set of constants of the motion become

$$\frac{1}{2}(p_x^2+p_y^2)-\frac{\mu}{r}+\frac{(1-k^2)L^2}{2r^2}=E,$$

$$\frac{(xp_x+yp_y)\cos\left(\frac{\theta}{2k}\right)-L\sin\left(\frac{\theta}{2k}\right)-i\sqrt{-2E}r\cos\left(\frac{\theta}{2k}\right)}{(xp_x+yp_y)\sin\left(\frac{\theta}{2k}\right)+L\cos\left(\frac{\theta}{2k}\right)-i\sqrt{-2E}r\sin\left(\frac{\theta}{2k}\right)}=-i\frac{\sqrt{1-e}}{\sqrt{1+e}}$$\vspace{-5mm}

\npgni and\vspace{-5mm}

$$\frac{1}{2}(p_x^2+p_y^2)+\frac{1}{2}\omega^2r^2+\frac{(1-k^2)L^2}{2r^2}=E,$$

$$\frac{(xp_x+yp_y)\cos\left(\frac{\theta}{k}\right)-L\sin\left(\frac{\theta}{k}\right)-i\omega r^2\cos\left(\frac{\theta}{k}\right)}{(xp_x+yp_y)\sin\left(\frac{\theta}{k}\right)+L\cos\left(\frac{\theta}{k}\right)-i\omega r^2\sin\left(\frac{\theta}{k}\right)}=-\frac{a}{b}i,$$\vspace{-8mm}

\npgni respectively.\vspace{-3mm}

\npgni Observing that each pair of simultaneous equations can be solved for $(p_x, p_y)\in\mathbb{C}^2$ it is natural to define $\vct Z=(Z_x,Z_y)\in\mathbb{C}^2$ and replace $\vct p$ by $\vct Z$ in the above equations. Solving each pair of equations leads to the solution\vspace{-3mm}

$$Z_x=(\alpha+i\beta)\frac{x}{r^2}-\frac{y}{r^2}\sqrt{\beta^2-\alpha^2+2r^2(E-V_0(r))-(1-k^2)L^2-2\alpha\beta i},$$\vspace{-10mm}

$$Z_y=(\alpha+i\beta)\frac{y}{r^2}+\frac{x}{r^2}\sqrt{\beta^2-\alpha^2+2r^2(E-V_0(r))-(1-k^2)L^2-2\alpha\beta i},$$\vspace{-10mm}

\npgni where, for:-\vspace{-5mm}

$$V_0=-\frac{\mu}{r},\;\alpha=\frac{eL\sin\left(\frac{\theta}{k}\right)}{1+e\cos\left(\frac{\theta}{k}\right)},\;\beta=\frac{1}{\lambda}\left(\mu r-\frac{L^2}{1+e\cos\left(\frac{\theta}{k}\right)}\right),\; L=\lambda\sqrt{1-e^2},\;E=-\frac{\mu^2}{2\lambda^2},$$\vspace{-8mm}

\npgni and for\vspace{-5mm}

$$V_0=\frac{1}{2}\omega^2r^2,\;\alpha=\frac{(b^2-a^2)L\sin\left(\frac{2\theta}{k}\right)}{(a^2+b^2)+(b^2-a^2)\cos\left(\frac{2\theta}{k}\right)},\;\beta=\omega r^2-\frac{2abL}{(a^2+b^2)+(b^2-a^2)\cos\left(\frac{2\theta}{k}\right)},$$\vspace{-5mm}

\npgni $L=ab\omega\textrm{ and }E=\dfrac{1}{2}(a^2+b^2)\omega^2.$

\npgni Moreover, if we define $\vct Z=\boldsymbol\nabla S-i\boldsymbol\nabla R$, as in the previous sections we can extract $\boldsymbol\nabla S$ and $\boldsymbol\nabla R$. In each case, for the appropriate $V_0$, $E$, $L$, $\alpha$ and $\beta$,

$$\boldsymbol\nabla S=\left(\frac{1}{r^2}(\alpha x-\tilde uy),\frac{1}{r^2}(\alpha y+\tilde ux)\right),$$

$$\boldsymbol\nabla R=\left(-\frac{1}{r^2}(\beta x-\tilde vy),-\frac{1}{r^2}(\beta y+\tilde vx)\right),$$

\npgni where $\tilde u+i\tilde v=\sqrt{\beta^2-\alpha^2+2r^2(E-V_0(r))-(1-k^2)L^2-2\alpha\beta i}$. Explicitly\vspace{-3mm}

\begin{multline*}
\tilde u= \left\{\frac{1}{2}(\beta^2-\alpha^2+2r^2(E-V_0(r))-(1-k^2)L^2)\right. \\ +\left.\frac{1}{2}\sqrt{(\beta^2-\alpha^2+2r^2(E-V_0(r))-(1-k^2)L^2)^2+4\alpha^2\beta^2}\right\}\tph
\end{multline*}\vspace{-12mm}

\npgni and

\npgni $\;\;\;\tilde v=-\dfrac{\alpha\beta}{\tilde u}$

\begin{multline*}
\tilde v= \left\{-\frac{1}{2}(\beta^2-\alpha^2+2r^2(E-V_0(r))-(1-k^2)L^2)\right. \\ +\left.\frac{1}{2}\sqrt{(\beta^2-\alpha^2+2r^2(E-V_0(r))-(1-k^2)L^2)^2+4\alpha^2\beta^2}\right\}\tph
\end{multline*}\vspace{-8mm}

\npgni Needless to say, it is easy to show that, in each case,

$$2^{-1}(|\boldsymbol\nabla S|^2-|\boldsymbol\nabla R|^2)+V=E,\spc\boldsymbol\nabla R.\boldsymbol\nabla S=0,$$\vspace{-8mm}

\npgni i.e. we have found solutions to the all important ‘semi- classical’ equations $(\ast)$.

\subsection{Semi-classical and Classical Orbits}\vspace{-5mm}

\npgni As defined previously the semi-classical orbit for these systems is governed by the dynamical equation\vspace{-3mm}

$$\frac{d}{dt}\vct X_t^0=(\boldsymbol\nabla S+\boldsymbol\nabla R)(\vct X_t^0).$$\vspace{-8mm}

\npgni We have also seen that by virtue of its construction the classical orbit is attained in the infinite time limit as $R$ reaches its maximum at which point $\boldsymbol\nabla R=\vct 0$. Using the expression for $\boldsymbol\nabla R$ above, this occurs when\vspace{-3mm}

$$\beta x-\tilde vy=0\textrm{ and }\beta y+\tilde vx=0.$$\vspace{-10mm}

\npgni Assuming $x$ and $y$ are not simultaneously 0 i.e. the orbit does not reach the origin in finite time, these can be guaranteed if $\beta^2+\tilde v^2=0\implies\beta=0\textrm{and} \tilde v=0$ simultaneously.

\npgni In the case of $V_0=-\dfrac{\mu}{r}$,\vspace{-5mm}

$$\beta=0\implies\frac{\ell}{r}=1+e\cos\left(\frac{\theta}{k}\right),$$\vspace{-8mm}

\npgni where $\ell=\dfrac{L^2}{\mu}=\dfrac{\lambda^2(1-e^2)}{\mu}$.

\npgni In the case of $V_0=\dfrac{1}{2}\omega^2r^2$,

$$\beta=0\implies r^2=\frac{2a^2b^2}{(a^2+b^2)+(b^2-a^2)\cos\left(\frac{2\theta}{k}\right)}$$\vspace{-5mm}

\npgni Using the expression for $\tilde v$ we also see that $\tilde v = 0$ when $\beta = 0$ i.e. on the classical orbit. These results beautifully display the result of Newton: if $r(\theta)$ describes the orbit under $V_0$ then $r\left(\dfrac{\theta}{k}\right)$ defines the orbit under $V_K$.

\subsection{The Significance of}\vspace{-8.5mm}

\hspace{44mm} $u\;and\;\alpha$\vspace{-3mm}

\npgni Using the dynamical equation, we see that the polar coordinates $(r,\theta)$ satisfy particularly simple
equations:\vspace{-4mm}

$$\frac{dr}{dt}=\frac{\alpha-\beta}{r},\spc\frac{d\theta}{dt}=\frac{\tilde u-\tilde v}{r^2}.$$\vspace{-7mm}

\npgni On the classical orbit $(\beta, \tilde v)=(0,0).$ In addition, $\tilde u=kL$, so that

$$\frac{d\theta}{dt}=\frac{kL}{r^2}\implies r^2\frac{d\theta}{dt}=kL$$

\npgni the angular momentum. From this we deduce that $\dfrac{kL}{r}\dfrac{dr}{d\theta}=\alpha$, which is easy to verify in both cases described above.

\npgni The figures above show some simulations of semi-classical orbits which clearly display the convergence to the classical paths.

\begin{figure}
    \centering
    \includegraphics[scale=0.85]{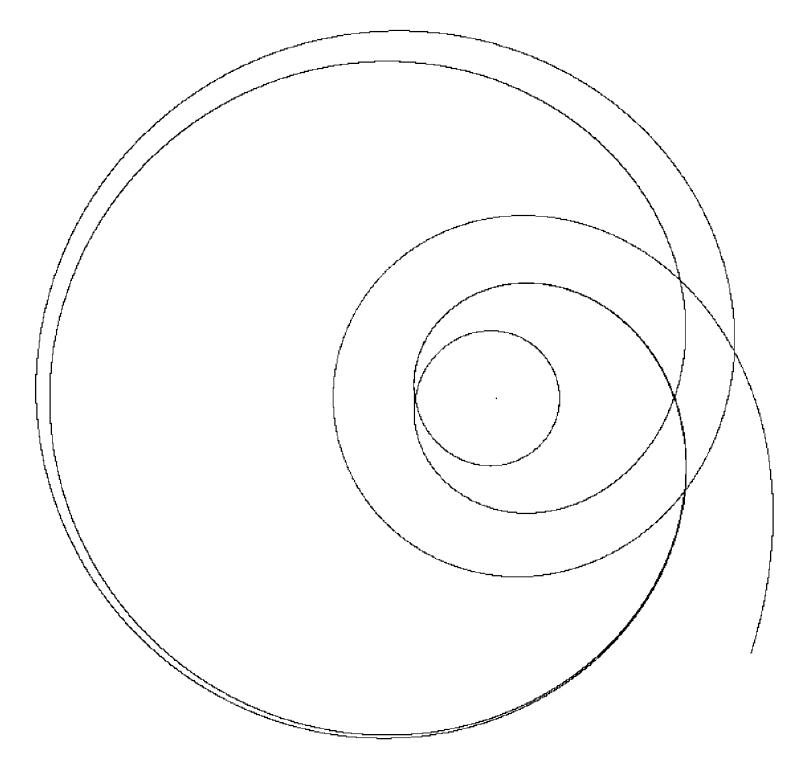}\spc\spc
    \includegraphics[scale=0.85]{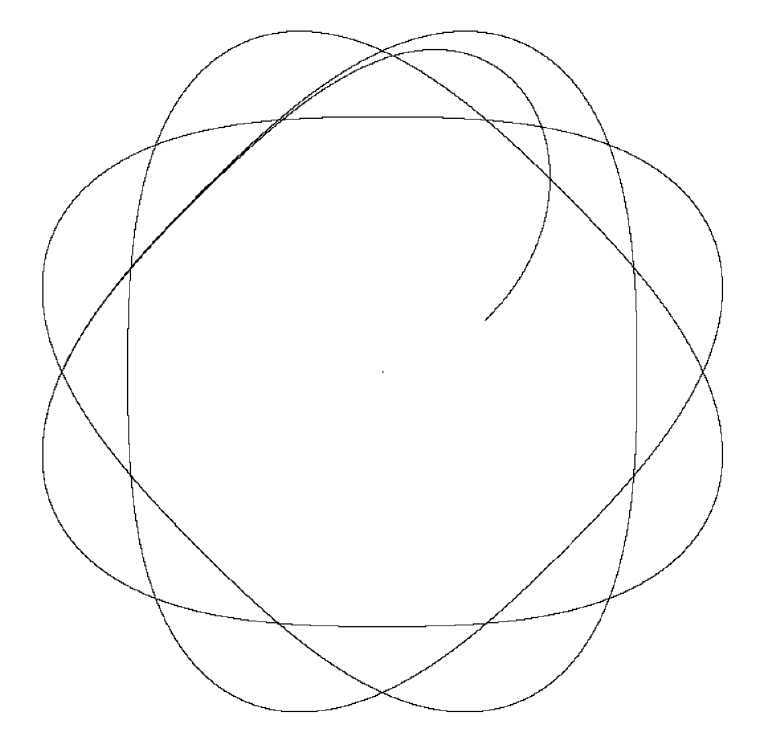}
    $$V_{0}=-\frac{\mu}{r}\;:\;k=3\spc\spc\spc\spc\spc\; V_{0}=\dfrac{1}{2}\omega^2r^2\;:\;k=\frac{3}{4}$$\vspace{-3mm}
    Figure 3.
\end{figure}\vspace{-3mm}

\npgni We conclude this section with an application of revolving orbits which comes from general relativity. Here a simple correction to the Newtonian force law accounts for the perihelion precession of a planet, such as Mercury, orbiting a star. The simplest model gives the corrected potential $V(r)=-\dfrac{\mu}{r}-\dfrac{3\mu^{2}}{c^{2}r^{2}}$, where $c$ is the speed of light. We can identify this with a revolving orbit, with $k=\sqrt{1+\dfrac{6\mu^{2}}{c^{2}L^{2}}}$. It is then easy to deduce that the perihelion of the orbit precesses, approximately, by an angle of $\dfrac{6\pi{\mu}^{2}}{c^{2}L^{2}}$ for each turn of the orbit. In the setting of our semi-classical revolving orbits the convergence of an orbit of this type is illustrated in Figure 4 below.

\begin{center}
\includegraphics[scale=0.8]{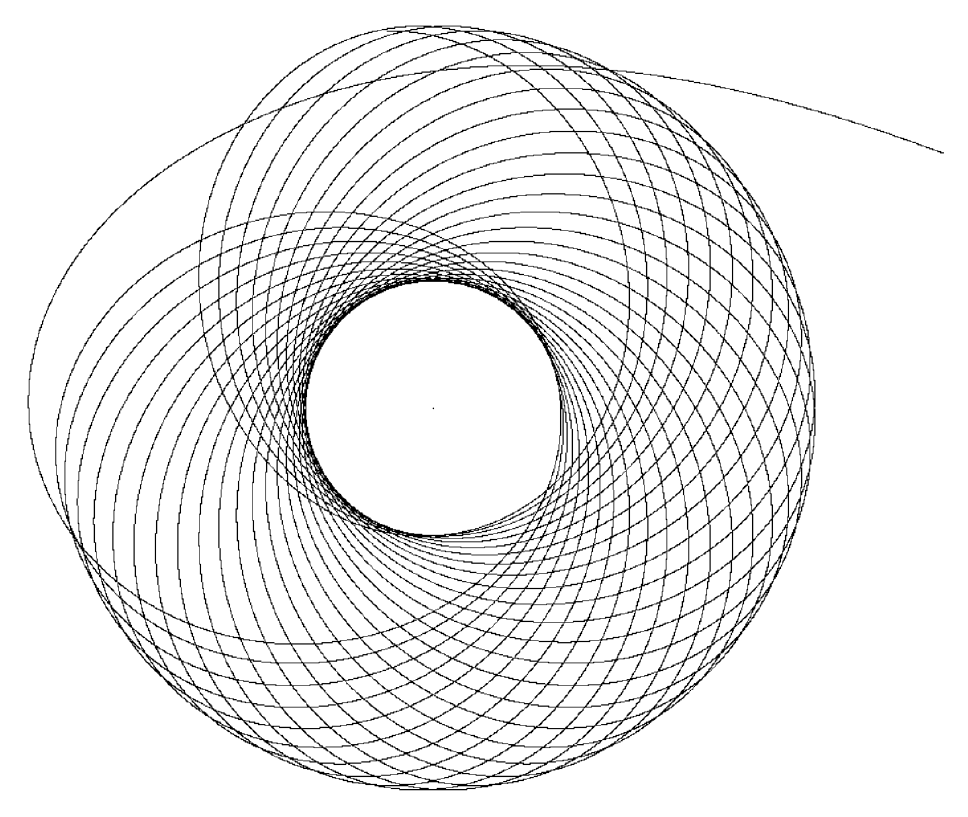}\\
Figure 4.
\end{center}

\section{Cauchy's Method of Characteristics and a Generalised LCKS Transformation in 2-dimensions}\vspace{-5mm}

\npgni Here we focus on the orthogonality of the level curves of $R$ and $S$,

$$\boldsymbol\nabla R.\boldsymbol\nabla S=0,$$\vspace{-8mm}

\npgni assuming $R$ is known exactly or approximately and that 
$S(x,y)=S_0(x,y),$ for $(x,y)\in\mathbb{C}_0$, curve of classical orbit, $S_0$ the Hamilton Jacobi function.

\subsection{Method of Characteristics}

\npgni Assume $\dfrac{\partial R}{\partial x}\neq0$ then we need

$$\frac{\partial S}{\partial x}+\dfrac{\dfrac{\partial S}{\partial y}\dfrac{\partial R}{\partial y}}{\dfrac{\partial R}{\partial x}}=0,$$\vspace{-7mm}

\npgni so if $y=y(x)$ and

$$\frac{dy}{dx}=\frac{\dfrac{\partial R}{\partial y}(x,y)}{\dfrac{\partial R}{\partial x}(x,y)},\spc \frac{d}{dx}S(x,y(x))=0.$$

\npgni So for each point $(x_0, y_0)\in\mathbb{C}_0$ define $\dfrac{dy}{dx}(x;x_0,y_0)$ as above, with\vspace{5mm}

$$\frac{dy}{dx}(x;x_0,y_0)=-\frac{\partial S_0}{\partial x}(x_0,y_0)\left/\frac{\partial S_0}{\partial y}(x_0,y_0)\right.$$

\npgni and $y(x_0;x_0,y_0)=y_0,\;(x_0,y_0)\in\mathbb{C}_0,$ then 

$$S(x,y(x;x_0,y_0))=S_0(x_0,y_0)\;\textrm{for each }(x_0,y_0)\in\mathbb{C}_0.$$\vspace{-8mm}

\npgni It follows that necessarily:

$$\frac{\partial S}{\partial x}+y'\frac{\partial S}{\partial y}=0,\spc\frac{\partial R}{\partial x}-\frac{1}{y'}\frac{\partial R}{\partial y}=0,$$\vspace{-8mm}

\npgni so that

$$|\boldsymbol\nabla S|^2=\left(\frac{\partial S}{\partial y}\right)^2(1+y'^2),\spc|\boldsymbol\nabla R|^2=\left(\frac{\partial R}{\partial y}\right)^2\left(1+\frac{1}{y'^2}\right).$$

\npgni Energy conservation then reduces to 

$$(S_y)^2\left(1+\frac{R_y^2}{R_x^2}\right)-(R_y)^2\left(1+\frac{S_y^2}{R_x^2}\right)=2(E-V),$$

\npgni $R_x=\dfrac{\partial R}{\partial x}$ etc. In this case we have an implicit solution of Eqs($\ast$). Eliminating $V$ we obtain:-

\begin{lemma}
\npgni A necessary condition for semi-classical Eqs($\ast$) to have a solution is that 

$$\left(1+\frac{S_y^2}{R_x^2}\right)|\boldsymbol\nabla R|^2=\left(1+\frac{R_y^2}{S_x^2}\right)|\boldsymbol\nabla S|^2$$

\npgni and modulo some mild regularity assumptions on $R$ and $S$, necessary and sufficient conditions are that additionally,

$$2(E-V)=|\boldsymbol\nabla R|^2\left(\frac{S_y^2}{R_x^2}-1\right).$$

\end{lemma}

\npgni \textbf{Remarks} 

\npgni 1. Sufficient regularity conditions can be deduced from Nagumo's results (Ref. [10]).

\npgni 2. Above results reveal the structure of the level curves of $R$ and $S$ in the neighbourhood of $C_0$ on which $R$ is known to have a maximum value and $S_0$ is known approximately. We return to this theme later. 

\npgni 3. It is easy to check the above conditions for all examples herein. 

\npgni 4. We need a more powerful result for explicit solutions. 

\subsection{Eplicit Solutions in 2-dimensions}

\npgni Where the potential $V$ is central, the classical angular momentum is conserved and the resulting motion can be assumed to be confined to the plane $z=0$. We work in $\mathcal{N}$, a 2-dimensional neighbourhood of an arc of $C_{0}$, the classical orbit, the neighbourhood being assumed to be simply connected. To be specific, $\mathcal{N}=\!\!\!\!\displaystyle\bigcup_{\theta\in I, \delta(\theta)>0}\!\!\!\!\!N((r_{0}(\theta),\theta),\delta(\theta))$, in polars, $(r_{0}(\theta),\theta)$ being the centre and $\delta(\theta)>0$ being the radius of $N$, $(r_{0}(\theta),\theta)\in C_{0}$, and $\theta\in I$, a real interval. We assume that the quantum particle density of state, $\psi$, satisfies, $\rho(x,y)\sim C\exp\left(\frac{2R(x,y)}{\epsilon^{2}}\right)$, as $\epsilon^{2}\sim0$, the real valued $R$ achieving its global maximum on $C_{0}$, where $\boldsymbol\nabla R=0$ and $R=R_{\textrm{max}}$.

\npgni To satisfy our equations ($\ast$) we write $\boldsymbol\nabla S(x,y)=\lambda(x,y)\left(-\dfrac{\partial R}{\partial y},\dfrac{\partial R}{\partial x}\right)$, where \\$\lambda(x,y)\sim O(|\boldsymbol\nabla R|^{-1})$ as $(x,y)\sim(x_{0},y_{0})\in C_{0}$. Assuming at all points of the arc of $C_{0}$ we have a continuously turning tangent $\vct t(C_{0})$, with $|\vct t(C_{0})|=1$ it is necessary that

$$\boldsymbol\nabla S(x,y)\sim \lambda(x,y)|\boldsymbol\nabla R|(x,y)\vct t(C_{0})(x_{0},y_{0})$$

\npgni as $(x,y)\sim (x_{0},y_{0})$. In any case

$$\frac{\partial S}{\partial x}=-\lambda\frac{\partial R}{\partial y},\;\;\;\;\;\frac{\partial S}{\partial y}=\lambda\frac{\partial R}{\partial x}$$\vspace{-7mm}

\npgni and assuming $C^2$ behaviour

$$\frac{\partial^{2} S}{\partial x\partial y}=\frac{\partial^{2} S}{\partial y\partial x}=-\frac{\partial\lambda}{\partial y}\frac{\partial R}{\partial y}-\lambda\frac{\partial^2R}{\partial y^{2}}=\frac{\partial\lambda}{\partial x}\frac{\partial R}{\partial x}+\lambda\frac{\partial^{2} R}{\partial x^{2}},$$

\npgni i.e., $\lambda$ has to satisfy the $(R,V)$ equation:

$$\lambda\Delta R=-\boldsymbol\nabla\lambda.\boldsymbol\nabla R$$

\npgni and by energy conservation $\lambda=\pm\sqrt{\left(1+\dfrac{2(E-V)}{|\boldsymbol\nabla R|^{2}}\right)}$.

\npgni Defining for $C(\vct{r}_{0},\vct{r})$ a simple curve from $\vct{r}_{0}\in C_{0}$ to $\vct{r}\in \mathcal{N}$, the line integral,

$$S(x,y)-S(x_{0},y_{0})=\int_{C(\vct{r}_{0},\vct{r})}\left(-\lambda\frac{\partial R}{\partial y},\lambda\frac{\partial R}{\partial x}\right).d\vct{r}$$\vspace{-5mm}

\npgni predicated on $\lambda$ satisfying the $(R,V)$ equation gives a unique $S$ defined in $\mathcal{N}$ with

$$\boldsymbol\nabla S=\pm\sqrt{|\boldsymbol\nabla R|^{2}+2(E-V)}(\widehat{\boldsymbol\nabla R^{\bot}}),\;\;\textrm{off}\;\;C_{0},$$

\npgni where $\boldsymbol\nabla R^{\bot}=\left(-\dfrac{\partial R}{\partial y},\dfrac{\partial R}{\partial x}\right)$, $(\widehat{\boldsymbol\nabla R^{\bot}})$ being $\vct t(C_{0})$ on the classical orbit.

\npgni Evidently as required $S\rceil_{C_0}=S_0$, the classical Hamilton Jacobi function.

\npgni We have proved modulo some mild regularity conditions:-\vspace{5mm}

\begin{theorem}

Assuming, $\lambda=\pm\sqrt{\left(1+\dfrac{2(E-V)}{|\boldsymbol\nabla R|^{2}}\right)}$, satisfies the $(R,V)$ equation\vspace{5mm}

$$\lambda\Delta R=-\boldsymbol\nabla\lambda.\boldsymbol\nabla R,$$

\npgni i.e. the flow with vector field $\lambda\boldsymbol\nabla R$ is incompressible and has no sources or sinks,

\npgni the solution to our equations ($\ast$), in 2-dimensions, for a given $R$, is

$$\boldsymbol\nabla S=\pm\sqrt{|\boldsymbol\nabla R|^{2}+2(E-V)}(\widehat{\boldsymbol\nabla R^{\bot}}).$$

\end{theorem}

\npgni In particular in a neighbourhood of $C_0$, the classical orbit, on the level surface,

\begin{multline*}
R=R_{0}=\left.\Big\{(r,\theta):r=r_{0}(\theta)+r_{1}(\theta),\;\frac{r_{1}^{2}}{2}\frac{\partial^{2}R}{\partial r^{2}}(r_{0},\theta_{0})=R_{0}-R_{\textrm{max}},\right. \\  \left.C:r=r_{0}(\theta),\theta\in(0,2\pi),\;\textrm{say}\;as\;R\sim R_{\textrm{max}}\right.\Big\},
\end{multline*}\vspace{-10mm}

\npgni we obtain the leading term,

$$\boldsymbol\nabla S\sim\pm\sqrt{2(E-V)}\left(1+\frac{|\boldsymbol\nabla R|^{2}}{4(E-V)}\right)\vct t(R=R_{0}),\;\textrm{as}\;R\sim R_{\textrm{max}},$$\vspace{-5mm}

\npgni $\vct t(R=R_{0})$ being the unit tangent to the level surface (contour), $R=R_{0}$, and tangent to $C_{0}$ if $R_{0}=R_{\textrm{max}}$.

\npgni This gives detailed information on the behaviour of the semi-classical limit in a narrow strip containing $C_{0}$ in terms of $R$, which we assume is known, at least as far as $\dfrac{\partial^{2}R}{\partial r^{2}}(r_{0},\theta)$, $\theta\in(0,2\pi)$ and $|\boldsymbol\nabla R|^{2}$, in $N$.

\begin{lemma}
Let $R(x,y)=R_{0}$ be a level surface of $R$ in a neighbourhood, $\mathcal{N}$, of the classical orbit $C_{0}$, with polar equation, $r=r_{0}$, $0<\theta\le2\pi$, then

$$R(x,y)\sim R_{\textrm{max}}+\frac{(r-r_{0}(\theta))^{2}}{2}\frac{\partial^{2}R}{\partial r^{2}}(r_{0}(\theta),\theta)+O\left((r-r_{0}(\theta))^{3}\right).$$

\npgni and the leading term for $|\boldsymbol\nabla R|^{2}$ is given by,

$$|\boldsymbol\nabla R|^{2}\sim 2(R_{0}-R_{\textrm{max}})\frac{\partial^{2}R}{\partial r^{2}}(r_{0},\theta)\left(1+\frac{(r_{0}'(\theta))^{2}}{r^{2}}\right).$$
\end{lemma}

\begin{proof}
Modulo sufficient regularity,

$$\frac{\partial R}{\partial r}\sim (r-r_{0}(\theta))\frac{\partial^{2}R}{\partial r^{2}}(r_{0}(\theta),\theta)+ \frac{(r-r_{0}(\theta))^{2}}{2}\frac{\partial^{3}R}{\partial r^{3}}(r_{0}(\theta),\theta)$$\vspace{-3mm}
\npgni and\vspace{-3mm}

\npgni $\spc\dfrac{1}{r}\dfrac{\partial R}{\partial \theta}\sim -\dfrac{(r-r_{0})}{r}r_{0}'(\theta)\dfrac{\partial^{2}R}{\partial r^{2}}(r_{0}(\theta),\theta)+\dfrac{ (r-r_{0}(\theta))^{2}}{2r}\dfrac{\partial}{\partial\theta}\left(\dfrac{\partial^{2}R}{\partial r^{2}}(r_{0}(\theta),\theta)\right)$

\npgni $\spc\spc\spc\spc\spc\spc\spc\spc\spc\spc-\dfrac{(r-r_{0}(\theta))^{2}}{2r}r_{0}'(\theta)\dfrac{\partial^{3}R}{\partial r^{3}}(r_{0}(\theta),\theta).$

\npgni Since $R=R_{\textrm{max}}$ on $C_{0}$, a simple computation gives the desired result.
\end{proof}\vspace{-5mm}

\npgni In 2-dimensions the formula for $\boldsymbol\nabla S$, apart from the obvious change of direction, merely includes a $|\boldsymbol\nabla R|^{2}$ term, with $V\rightarrow V-\dfrac{|\boldsymbol\nabla R|^{2}}{2}$ and

$$-\dfrac{|\boldsymbol\nabla R|^{2}}{2}=\frac{(R_{\textrm{max}}-R_{0})}{r^{2}}\frac{\partial^{2}R}{\partial r^{2}}(r_{0},\theta)\left(r^{2}+(r_{0}'(\theta))^{2}\right).$$\vspace{-10mm}

\npgni Since\vspace{-3mm}

$$\frac{d}{dt}\vct X_t^0=(\boldsymbol\nabla R+\boldsymbol\nabla S)(\vct X_t^0),$$

\npgni we see that as $R_{0}\sim R_{\textrm{max}}$ on the level surface $R=R_{0}$,

$$h-h_{0}\cong\frac{(R_{\textrm{max}}-R_{0})}{r}\frac{\partial^{2}R}{\partial r^{2}}(r_{0},\theta)\left(r^{2}+(r_{0}'(\theta))^{2}\right)+\frac{\partial R}{\partial\theta}(r,\theta),$$\vspace{-5mm}

\npgni where $r_{0}'(\theta)=\dfrac{dr_{0}}{d\theta}$, showing how the shape of the classical orbit $r=r_{0}(\theta)$ affects the quantum correction to angular momentum, $h_{0}$ being the classical angular momentum. Needless to say for central potentials $V$, $h_{0}$ is a constant.\vspace{2mm}

\begin{corollary}
Modulo the $(R,V)$ equation being satisfied and some mild regularity, in general, off the classical orbit $C_{0}$,
$$r^{2}\dot{\theta}=\left.\left(\frac{\partial R}{\partial\theta}(r,\theta)+r\sqrt{|\boldsymbol\nabla R|^{2}+2(E-V)}\left({|\boldsymbol\nabla R|}^{-1}\frac{\partial R}{\partial r}\right)\right)\right\rceil_{R=R_{0}},$$\vspace{-8mm}

\npgni with above interpretation of the bracketed second term on $C_{0}$, where we reiterate we assume $R$ is known. Similarly,

$$\dot{r}=\left.\left(\frac{\partial R}{\partial r}(r,\theta)+\frac{\sqrt{|\boldsymbol\nabla R|^{2}+2(E-V)}}{r}\left({|\boldsymbol\nabla R|}^{-1}\frac{\partial R}{\partial \theta}\right)\right)\right\rceil_{R=R_{0}}.$$
\end{corollary}\vspace{-2mm}

\npgni We conclude this section with the related result for circular orbits.

\begin{theorem}
For the spiral circular orbit, $C_{0}$, in any central potential $V$, solving the $(R,V)$ equation in polars $(r,\theta)$ yields:-

$$R+iS=f(r)+ih\theta,$$\vspace{-8mm}

\npgni where, for $(')=\dfrac{d}{dr}$,

$$r^{2}f'^{2}(r)=4\int^{r}_{a}\rho^{2}\left(\frac{V'(\rho)}{2}+\frac{V(\rho)-E}{\rho}\right)d\rho,$$

\npgni with $\dfrac{V'(a)}{2}+\dfrac{V(a)-E}{a}=0$, $a$ being the radius of the circular orbit and $E$ and $h$, the\vspace{3mm} energy and angular momentum respectively. Moreover

$$(f''(a))^{2}=-\frac{d}{dr}\left.\left(\frac{2(E-V(r))}{r}-V'(r)\right)\right\rceil_{r=a}=V''(a)+\frac{3h^{2}}{a^{4}},$$

\npgni so the particle density on $C_{0}$,

$$\rho\sim C_{\epsilon}\textrm{exp}\left(\frac{1}{\epsilon^{2}}\left(2R_{\textrm{max}}-\left(V''(a)+\frac{3h^{2}}{a^{4}}\right)^{\frac{1}{2}}(r-a)^{2}\right)\right),$$

\npgni provided $V''(a)+\dfrac{3h^{2}}{a^{4}}>0$, i.e. the classical condition for a stable circular orbit. $C_{\epsilon}$ is\vspace{3mm} determined by normalisation.
\end{theorem}

\begin{proof}
\npgni Given $R+iS=f(r)+ih\theta,$ the $(R,V)$ equation implies\vspace{5mm}

$$\frac{1}{r}(rf'(r))'+\frac{1}{2}\boldsymbol\nabla\ln\left(1+\frac{2(E-V(r))}{f'^{2}(r)}\right).f'(r)\hat{\vct{r}}=0$$\vspace{-8mm}

\npgni i.e.\vspace{-5mm}

$$f''f'^{3}+f'^{2}\left(\frac{f'^{2}+2(E-V)}{r}-V'\right)=0$$

\npgni which can be integrated, with $f'(a)=0$, to give\vspace{5mm}

$$r^{2}f'^{2}(r)=4\int^{r}_{a}u^{2}\left(\frac{V'(u)}{2}+\frac{V(u)-E}{u}\right)du.$$

\npgni Since we need a repeated root at $r=a$ we require

$$\dfrac{V'(a)}{2}+\dfrac{V(a)-E}{a}=0$$

\npgni as well as the classical condition $\dfrac{d}{dr}\left.\left(V+\dfrac{h^{2}}{2r^{2}}\right)\right\rceil_{r=a}=0$, where $h$ is the angular momentum. Moreover, for $f'(r)\neq0$,

$$f''=-\frac{\left(\dfrac{f'^{2}+2(E-V)}{r}-V'\right)}{f'}.$$

\npgni An application of de l’Hopital's rule for $r\rightarrow a$ gives

$$(f''(a))^{2}=-\frac{d}{dr}\left.\left(\frac{2(E-V(r))}{r}-V'(r)\right)\right\rceil_{r=a}=V''(a)+\frac{3h^{2}}{a^{4}},$$

\npgni and the result for the particle density follows easily.

\end{proof}

\npgni An example for the above is provided by motion in the potential field

$$V(r)=-\dfrac{\mu}{r}+\dfrac{1}{2}\omega^{2}r^{2}+\dfrac{\omega\mu}{h}r$$

\npgni in 2-dimensions, for positive $\mu$, $\omega$ and $h$. This represents a Coulomb (neutron star/black hole), harmonic oscillator and dark energy potential. The exact solution to our equations ($\ast$), in this setting, is

$$R+iS=f(r)+iS=-\dfrac{1}{2}\omega r^{2}-\dfrac{\mu}{h}r+h\ln r+ih\theta.$$

\npgni It is easy to check that $f'(r)=0$ implies a circular orbit at

$$a=\dfrac{\sqrt{\mu^{2}+4\omega h^{3}}-\mu}{2\omega h}$$

\npgni and that $(f''(a))^{2}=V''(a)+\dfrac{3h^{2}}{a^{4}}=\omega^{2}+\dfrac{h^{2}+2\omega ha^{2}}{a^{4}}>0$, giving a stable circular orbit.

\npgni A computer simulation for the corresponding semi-classical motion is shown in Figure 6 at the end of this paper.

\subsection{Power Law Problem for Zero Energy in 2-dimensions}\vspace{-5mm}

\npgni Quite recently the zero energy power law problem has attracted the attention of several researchers. An excellent survey is given in the paper by A.J. Makovski and K.J. Gorska (Ref. [14]). In our setting it reduces to solving the Eqs($\ast$) for $V(r)=-\kappa r^p$, most importantly,\vspace{-5mm} 

$$2^{-1}(|\boldsymbol\nabla S|^2-|\boldsymbol\nabla R|^2)-\kappa r^p=0,\spc\boldsymbol\nabla R.\boldsymbol\nabla S=0,\spc\boldsymbol\nabla=\left(\frac{\partial}{\partial x},\frac{\partial}{\partial y}\right),$$\vspace{-5mm}

\npgni where $\boldsymbol\nabla=\boldsymbol\nabla_z$, $z=x+iy$, $r=\sqrt{(x^2+y^2)}=|z|.$ We linearise this problem by changing variables to $\boldsymbol\nabla_Z$, $Z=X+iY$, $Z=(z)^{(p+2)/2}$, then $f(z)=z^{(p+2)/2}$ 

$$|\boldsymbol\nabla_z f|^2=|f'(z)|^2|\boldsymbol\nabla_Z f|^2$$\vspace{-10mm}

\npgni giving

$$2^{-1}(|\boldsymbol\nabla_Z S|^2-|\boldsymbol\nabla_Z R|^2)-\kappa|z|^p|f'(z)|^{-2}=0$$\vspace{-10mm}

\npgni i.e.\vspace{-5mm}

$$2^{-1}(|\boldsymbol\nabla_Z S|^2-|\boldsymbol\nabla_Z R|^2)-\kappa|z|^p\left(\frac{p+2}{2}\right)^{-2}|z|^{-p}=0$$\vspace{-8mm}

\npgni i.e.\vspace{-5mm}

$$2^{-1}(|\boldsymbol\nabla_Z S|^2-|\boldsymbol\nabla_Z R|^2)=\frac{4\kappa}{(p+2)^2}$$

\npgni Writing $R=\dfrac{\beta \tilde\mu}{\sqrt{\alpha^2 +\beta^2}}(\beta Y+\alpha X),\:S=\dfrac{\alpha \tilde\mu}{\sqrt{\alpha^2+\beta^2}}(-\beta X+\alpha Y),$ for constants $\alpha,\beta,\tilde\mu,$

\npgni for cartesian $X$ and $Y$,

$$\boldsymbol\nabla R=\frac{\beta\tilde\mu}{\sqrt{\alpha^2 +\beta^2}}(\alpha,\beta),\spc\boldsymbol\nabla S=\dfrac{\alpha\tilde\mu}{\sqrt{\alpha^2+\beta^2}}(-\beta,\alpha)$$

\npgni So, if $\tilde\mu=2\sqrt{2\kappa}/(p+2),$ $\alpha^2-\beta^2=1$,

$$\boldsymbol\nabla R.\boldsymbol\nabla S=0,\spc2^{-1}(|\boldsymbol\nabla S|^2-|\boldsymbol\nabla R|^2)=2^{-1}\tilde\mu^2(\alpha^2-\beta^2)=\dfrac{4\kappa}{(p+2)^2}$$

\npgni Hence, in the original variables $(x,y)=(r\cos\theta,r\sin\theta),$

$$X=r^{(p+2)/2}\cos\left(\frac{(p+2)}{2}\theta\right),\spc Y=r^{(p+2)/2}\sin\left(\frac{(p+2)}{2}\theta\right)$$\vspace{-10mm}

\npgni and 

$$R=\frac{\beta\tilde\mu}{\sqrt{\alpha^2+\beta^2}}(\beta Y+\alpha X),\spc S=\frac{\alpha \tilde\mu}{\sqrt{\alpha^2+\beta^2}}(-\beta X+\alpha Y),$$

\npgni solving Eqs($\ast$). The last example is the inspiration for the next section. 

\subsection{Holomorphic Change of Variables of Semi-Classical LCKS Transformation}

\npgni Assume $\omega=\xi+i\eta=f(x+iy)$, $z=x+iy$, $f$ being holomorphic and invertible for $z\in\mathbb{D}\subset\mathbb{C};\xi,\eta,x,y\in\mathbb{R}.$ We write:- $x=x(\xi,\eta),y=y(\xi,\eta);\xi=\xi(x,y),\eta=\eta(x,y)$, 

\npgni and assume $\dbinom{\dot x}{\dot y}=\boldsymbol\nabla_z(R+S)=\dbinom{\partial_x(R+S)}{\partial_y(R+S)}.$

\npgni Changing variables for any $C^1$ function $g(x,y)$, $(x,y)\to(\xi,\eta)$,

$$\begin{pmatrix}\dfrac{\partial g}{\partial x}\\[0.75em]\dfrac{\partial g}{\partial y}\end{pmatrix}=\begin{pmatrix}\dfrac{\partial \xi}{\partial x} & \dfrac{\partial \eta}{\partial x} \\[0.75em]  \dfrac{\partial\xi}{\partial y} & \dfrac{\partial \eta}{\partial y}\end{pmatrix}\begin{pmatrix}{\dfrac{\partial g}{\partial \xi}}\\[0.75em] \dfrac{\partial g}{\partial \eta}\end{pmatrix}$$\vspace{-5mm}

\npgni i.e.\vspace{-5mm}

$$\boldsymbol\nabla_z g=\begin{pmatrix}\dfrac{\partial \xi}{\partial x} & \dfrac{\partial \eta}{\partial x} \\[0.75em]  \dfrac{\partial\xi}{\partial y} & \dfrac{\partial\eta}{\partial y}\end{pmatrix}\boldsymbol\nabla_w g.$$

\npgni From the Cauchy Riemann equations

$$\boldsymbol\nabla_z g.\boldsymbol\nabla_z h=|f'(z)|^2\boldsymbol\nabla_w g.\boldsymbol\nabla_w h,$$

\npgni for any $C^1$ functions $g, h$. So for our functions $R, S,$

$$\boldsymbol\nabla_z R.\boldsymbol\nabla_z S=|f'(z)|^2\boldsymbol\nabla_w R.\boldsymbol\nabla_w S$$\vspace{-12mm}

\npgni and\vspace{-5mm}

$$2^{-1}(|\boldsymbol\nabla_z S|^2-|\boldsymbol\nabla_z R|^2)+(V-E)=2^{-1}(|\boldsymbol\nabla_w S|^2-|\boldsymbol\nabla_w R|^2)|f'(z)|^2+(V-E).$$\vspace{-8mm}

\npgni So, assuming $f'(z)\neq0$,

$$\boldsymbol\nabla_w R.\boldsymbol\nabla_w S=0,\spc 2^{-1}(|\boldsymbol\nabla_w S|^2-|\boldsymbol\nabla_w R|^2)+|f'(z)|^{-2}(V-E)=0.$$

\npgni Now, if $\dbinom{\dot x}{\dot y}=\boldsymbol\nabla_z(R+S)$, $(\raisebox{4pt}{$\cdot$})=\dfrac{d}{dt}$, $t$ being the time, since

$$\dbinom{\dot x}{\dot y}=\begin{pmatrix}\dfrac{\partial x}{\partial \xi} & \dfrac{\partial x}{\partial \eta} \\[0.75em] \dfrac{\partial y}{\partial \xi} & \dfrac{\partial y}{\partial \eta}\end{pmatrix}\binom{\dot\xi}{\dot\eta}$$\vspace{1mm}

$$\binom{\dot\xi}{\dot\eta}=\begin{pmatrix}\dfrac{\partial x}{\partial \xi} & \dfrac{\partial x}{\partial \eta} \\[0.75em]  \dfrac{\partial y}{\partial \xi} & \dfrac{\partial y}{\partial \eta}\end{pmatrix}^{-1}\binom{\dot x}{\dot y}=\begin{pmatrix} \dfrac{\partial\xi}{\partial x} & \dfrac{\partial\xi}{\partial y} \\[0.75em] \dfrac{\partial\eta}{\partial x} & \dfrac{\partial\eta}{\partial y} \end{pmatrix}\begin{pmatrix}\dfrac{\partial\xi}{\partial x} & \dfrac{\partial\eta}{\partial x} \\[0.75em] \dfrac{\partial\xi}{\partial y} & \dfrac{\partial \eta}{\partial y}\end{pmatrix}\boldsymbol\nabla_w(R+S).$$

\npgni From the Cauchy-Riemann equations

$$\binom{\dot \xi}{\dot\eta}=|f'(z)|^2\begin{pmatrix}{\dfrac{\partial}{\partial \xi}(R+S)}\\[0.75em] \dfrac{\partial}{\partial\eta}(R+S)\end{pmatrix},$$

\npgni so, if $ds=dt|f'(z(t))|^2$ is the time-change and $\dbinom{\xi(s)}{\eta(s)}$ inherits the corresponding initial conditions as $\dbinom{x(t)}{y(t)}$, $\dfrac{d}{ds}=(')$ we obtain 

$$\dbinom{\xi'}{\eta'}=\boldsymbol\nabla_w(R+S),$$

\npgni as expected. This follows from the obvious theorem for $\mathbb{D}\subset\mathbb{C}$, an open simply connected set.

\begin{theorem}
If $f(z)$ is a holomorphic, invertible function of $z\in\mathbb{D}$ with $|f'(z)|\neq0$ in $\mathbb{D}$, the transformation $\omega=f(z)$ transforms the Schrödinger equation in 2 dimensions, 

$$-2^{-1}\hbar^2\Delta_z\psi+(V-E)\psi=0,\spc z\in\mathbb{D},$$\vspace{-12mm}

\npgni to\vspace{-5mm}

$$-2^{-1}\hbar^2\Delta_w\psi+(V-E)|f'|^{-2}\psi=0,\spc w\in f(\mathbb{D}).$$
\end{theorem}

\begin{proof}
Cauchy-Riemann equations.
\end{proof}

\begin{corollary}
(Semi-classical LCKS tranformation)

\npgni Modulo above assumptions on $f$, $w=f(z)$ transforms Eqs($\ast$) to: 

$$2^{-1}(|\boldsymbol\nabla_w S|^2-|\boldsymbol\nabla_w R|^2)+|f'(z)|^{-2}(V-E)=0,\spc \boldsymbol\nabla_w R.\boldsymbol\nabla_w S=0$$

\npgni and if $s$ is the above time-change, $\dfrac{d}{ds}=(')$,\vspace{5mm}

$$\xi'=\frac{\partial}{\partial\xi}(R+S),\spc\eta'=\frac{\partial}{\partial\eta}(R+S),$$

\npgni with appropriate initial conditions. 

\npgni Further, if $R(x,y)=R_{\max}$ on $C_0$, and $|\boldsymbol\nabla R|=0$ on $C_0$, transformed $R$ inherits same properties guaranteeing convergence to corresponding classical orbit, $f(C_0)$, since

$$\frac{d}{ds}R(\xi,\eta)=\boldsymbol\nabla_w(R+S).\boldsymbol\nabla_w R=|\boldsymbol\nabla R|^2\ge0,$$

\npgni and, if the classical orbit is periodic, the time-change is valid since it will take infinite $s$ time to reach $f(C_0)$.
\end{corollary}

\npgni \textbf{Remarks}

\npgni 1. $(\xi,\eta)$ are orthogonal coordinates. We can regard new potential energy function as $W=|f'|^{-2}(V-E)$ and new energy is 0. So we can investigate the zero energy limit for the new potential $W$.

\npgni 2. It is now easy to check, $\omega=f(z)$, transforms constrained the Hamiltonian system $H=\dfrac{\vct{p}^2}{2}+V_{\textrm{eff}}(q)$, $(\vct{p}-\boldsymbol\nabla R).\boldsymbol\nabla R=0$, $V_{\textrm{eff}}=V-|\boldsymbol\nabla R|^2$, into the one corresponding with the above system.

\section {Newtonian Quantum Gravity and Deviations from Keplerian Motion}

The results in this section are heavily dependent on our work on constrained Hamiltonians underlying our asymptotic analysis in Refs. [15], [16] and [17]. These results are a simple consequence of taking the ultimate Bohr correspondence limit of the atomic elliptic state $\psi_{n,e}$, $\psi_{n,e}\sim\psi_{\textrm{sc}}=\textrm{exp}\left(\frac{R+iS}{\epsilon^2}\right)$ as $\epsilon\sim0$. They are predicated on the fact that $R$ attains its global maximum on the Kepler ellipse $\mathcal{E}_K$, in the plane $z=0$, with polar equation,

$$\frac{a\left(1-e^{2}\right)}{r}=1+e\cos\theta,$$

\npgni $a$ being the semi-major axis of the elliptical orbit, parallel to the x-axis, $e$ its eccentricity:

$$e = \sqrt{1+\frac{2L^{2}E}{\mu^{2}}},$$

\npgni where L is the orbital angular momentum, $E<0$ the total energy and $\mu$ is the gravitational mass at the force centre O. (see Ref. [6]). To simplify our results we need to assume the initial conditions:

$$R\left(\mathbf{X}_{t=0}^{0}\right)\geqslant M(\Sigma)=\lambda\left(2\left(\dfrac{1-e}{1+e}\right)-\ln4\right),\spc\lambda=(\mu a)^{1/2}.$$

\subsection {Convergence to the Plane}\vspace{-8mm}

\hspace{55mm} $z=0$

\npgni Recall that

$$\frac{\dot{z}}{z}=\frac{-\mu(\alpha+\beta+1)}{2\lambda r},\spc\ r=\sqrt{x^{2}+y^{2}+z^{2}}.$$

\npgni Since $\dot{z}=0$ when $z=0$, the orbiting particle takes an infinite time to reach the plane $z=0$, so we seek to understand the motion in this infinite time limit in which $u\rightarrow e$ and

$$\alpha \rightarrow\frac{\left(1-e^{2}\right)}{\left(1+e^{2}-2e\cos v\right)},\spc\beta\rightarrow -\frac{2e\sin v}{\left(1+e^{2}-2e\cos v\right)},$$

\npgni where $v \nearrow \infty$, $v$ being the eccentric anomaly on ellipse $\mathcal{E}_K$, with

$$\frac{dv}{dt}=\sqrt{\frac{\mu}{a^{3}}}\frac{1}{(1-e\cos v)}$$
and
$$\tan^{2}\left(\frac{v}{2}\right)=\frac{(1-e)}{(1+e)}\tan^{2}\left(\frac{\theta}{2}\right),$$

\npgni where $\theta$ is the polar angle of orbiting particle. Hence, we obtain

$$\frac{dz}{dt}\rightarrow - \frac{1}{2}\frac{(2-e^{2}+2e\sin v)}{(1+e^{2}-2e\cos v)} \spc \textrm {as} \spc v \nearrow \infty,$$

\npgni where we assume $z(0)>0$, giving the rate of convergence in the form:\vspace{5mm}

\begin{theorem}

$$z(v) \rightarrow C\left(1+e^{2}-2e\cos v\right)^{-1/2} \exp\left(-\frac{1}{2} \frac{\left(2-e^{2}\right)}{\left(1-e^{2}\right)} \tan^{-1} \left(\frac{1+e}{1-e}\tan \left(\frac{v}{2}\right)\right)\right),$$ \vspace{-5mm}

\npgni for some constant C, as $v \nearrow \infty$ for $0<e<1$.

\end{theorem}
\subsection {Convergence to Kepler Ellipse in 3-Dimensions}\vspace{-5mm}

\npgni The Keplerian coordinates $(u,v)$ come into their own in this setting:

$$x=\frac{2ae(\cos v-u)}{(e+u)}, \spc y=\frac{2ae\sqrt{\left(1-u^{2}\right)}\sin v}{(e+u)},\spc v\in \mathbb{R}.$$

\npgni The Kepler ellipse $\mathcal{E}_K$ corresponds to $u=e$, the singularity $\Sigma$ to $u=1$ and $u=-e$ is the curve at infinity.

$$\alpha=\frac{\sqrt{\left(1-u^{2}\right)\left(1-e^{2}\right)}}{(1+eu-(e+u)\cos v)},\spc \beta=-\frac{(e+u)\sin v}{(1+eu-(e+u)\sin v)},$$

$$b_{x}=\frac{\mu}{2\lambda e}\frac{\left(eu-(e-u)\cos v-(e+u)\sin v+\sqrt{\left(1-e^{2}\right)\left(1-u^{2}\right)}-1\right)}{(1-u\cos v)},$$

$$b_{y}=\frac{\mu}{2\lambda e}\frac{\left(\sqrt{1-u^{2}}(e\cos v-e\sin v+1)+\sqrt{1-e^{2}}(u\cos v+u\sin v-1)\right)}{(1-u\cos v)}$$\vspace{-8mm}

\npgni and the Jacobian

$$\Delta=\frac{\partial(x,y)}{\partial(u,v)}=\frac{4a^{2}e^{2}}{(e+u)^{3}\sqrt{1-u^{2}}}(u\cos v-1)((e+u)\cos v+1+eu).$$

\npgni In the infinite time limit we can take $v$ to be the eccentric anomaly on the Kepler ellipse, $\mathcal{E}_K$, itself. With this proviso we consider the ordinary differential equation

$$\frac{du}{dv}=\frac{b_{u}(u,v)}{b_{v}(u,v)},\spc \frac{du}{dt}=b_{u}(u,v),\spc\frac{dv}{dt}=b_{v}(u,v),$$

\npgni where in the limit\vspace{-3mm}

$$b_{v}=\sqrt{\frac{\mu}{a^{3}}}\frac{1}{(1-e\cos v)},$$

\npgni $b_{u}$ as given in section 2.3. We then have the simple result:\vspace{5mm}

\begin{theorem}
For $v>v_{0}$, $(v-v_{0})$ fixed as $v_{0} \nearrow \infty$,
$$\frac{u(v)-e}{u(v_{0})-e}\rightarrow \left(\frac{\cos v_{0}+\frac{1}{2}\left(e+e^{-1}\right)}{\cos v+\frac{1}{2}\left(e+e^{-1}\right)}\right)^{\frac{1}{2}}\exp\left(-(f(v)-f(v_{0}))\right),$$
\npgni where $f(v)=\dfrac{v}{2}+\tan^{-1}\left(\dfrac{1-e}{1+e}\tan\left(\dfrac{v}{2}\right)\right).$
\end{theorem}

\begin{proof}

Since $b_{u}(e,v)=0$,

$$\frac{d}{dv}\ln(u(v)-e)=\frac{\dfrac{du}{dv}}{(u-e)}=\frac{b_{u}(u,v)}{b_{v}(u,v)(u-e)}\rightarrow \frac{\dfrac{\partial b_{u}}{\partial u}(e,v)}{b_{v}(e,v)},$$

\npgni giving

$$\frac{d}{dv}\ln(u(v)-e)\rightarrow - \frac{1}{2}\frac{(1+e\cos v-e\sin v)}{(\cos v+\frac{1}{2}(e+e^{-1}))},$$

\npgni from which the result follows.

\end{proof}\vspace{-5mm}

\subsection {Transition to the Classical Era}

\npgni The cognoscenti will regard the above as a ‘half Nelson’. To them it will be important to repeat the above argument but include the Nelson noise coming from stochastic mechanics (Refs. [18],[19],[20],[21]). In this set up for orbiting particles of mass $m_{0}$, $\epsilon^{2}=\frac{\hbar}{m_{0}}$ and we seek the limit as $m_{0}\nearrow \infty$. Here we model this situation by assuming $\epsilon =\epsilon (u)\searrow 0$ as $u \rightarrow e$, as a result of mass accretion. In the last section $u(v)$, the solution of the above ordinary differential equation, corresponds to $u^{\epsilon}$ when $\epsilon=0$ where late in the semi-classical era $u^{\epsilon}$ satisfies, as $\epsilon \sim 0$,

$$du^{\epsilon}=b_{u}(u^{\epsilon},v)dt+\epsilon dN_{u}(t),$$

$$\frac{dv}{dt}=\sqrt{\frac{\mu}{a^{3}}}\frac{1}{(1-e\cos v)},\spc N_{u}(t),\;\;\textrm {Nelson noise}.$$

\npgni So here we denote $u(v)$ as $u^{0}(v)$.

\npgni A simple calculation yields for $\vct{r}^{\bot}_{v}=\left(\dfrac{\partial y}{\partial v},-\dfrac{\partial x}{\partial v}\right)$\vspace{5mm}

$$dN_{u}(t)=\Delta^{-1}(\vct{X}^{0}_{t})\vct{r}^{\bot}_{v}.d\vct{B}(t),$$

\npgni where $\vct{B}(t)=(B_{x},B_{y}),$ is a BM($\mathbb{R}^{2}$) process.

\npgni A simple time-change to the elliptic anomaly $v$ gives

$$dN_{u}(v)=\frac{(1-e\cos v)^{1/2}}{a^{2}e}\left(\frac{a^{3}}{\mu}\right)^{1/4}\frac{\left(a \sqrt{1-e^{2}}\cos v dB_{x}(v)+a\sin v dB_{y}(v)\right)}{(\cos v-e^{-1})(\cos v+\frac{1}{2}(e+e^{-1}))}.$$

\npgni So for $W$, a BM($\mathbb{R}$) process, correlated to $B_{x}$ and $B_{y}$,

$$du^{\epsilon}(v)=\frac{b_{u}}{b_{v}}dv+\frac{\epsilon(u^{\epsilon})\sqrt{1-e^{2}}\sqrt{1+e\cos v)}}{(a\mu)^{1/4}(\cos v+\frac{1}{2}(e+e^{-1}))}dW(v).$$

\npgni Assuming $\epsilon(u=e)=0$, letting $v\nearrow \infty$ gives

$$d\ln \left(\frac{u^{\epsilon}-e}{u^{0}-e}\right)=\frac{\epsilon^{\prime}(e)\sqrt{1-e^{2}}\sqrt{1-e\cos v}}{(a\mu)^{1/4}(\cos v+\frac{1}{2}(e+e^{-1}))}dW(v).$$

\npgni So, if $u^{\epsilon}$ and $u^{0}$ agree initially in this era,

$$u^{\epsilon}(v)=(u^{0}(v)-e)\exp\left(\frac{\epsilon^{\prime}(e)\sqrt{1-e^{2}}}{(a\mu)^{1/4}}W(g(v))\right),$$\vspace{-8mm}

\npgni where

$$g(v)=\int_{v_{0}}^{v} \dfrac{(1+e\cos v)}{(\cos v +\frac{1}{2}(e+e^{-1}))^{2}}dv.$$

\npgni For small eccentricities, $e\sim 0$, we have the log normal result:

$$\ln \left(\frac{u^{\epsilon}-e}{u^{0}-e}\right) \sim \frac {2e \epsilon^{\prime}(e)}{(a\mu)^{1/4}}W(v)\spc \textrm{as} \spc v\sim\infty,$$\vspace{-5mm}

\npgni where $v$ is the eccentric anomaly on the Kepler ellipse $\mathcal{E}_{u=e}$.

\npgni More generally as $v\sim\infty$,

$$\ln \left(\frac{u^{\epsilon}(v)-e}{u^{0}(v)-e}\right) \sim \frac {2e \epsilon^{\prime}(e)}{\sqrt{1-e^{2}}(a\mu)^{1/4}}W\left(\frac{2}{(1-e^{2})}\tan^{-1}\left(\frac{1-e}{1+e}\tan\frac{v}{2}\right)+\frac{e\sin v}{(1+2e\cos v+e^{2})}\right),$$\vspace{-8mm}

\npgni and\vspace{-3mm}

$$\left[\sqrt{\frac{\mu}{a^{3}}}(v^{\epsilon}(v)-v)\right]^{\beta}_{\alpha}\sim \int^{\beta}_{\alpha}\frac{\partial b_{v}}{\partial u}(e,v)(u^{\epsilon}(v)-e)(1-e\cos v)dv,$$

\npgni where $v^{\epsilon}$ is the eccentric anomaly on $\mathcal{E}_{u^{\epsilon}}$.

\npgni These are Nelsonian quantum corrections to Kepler's laws to first order of approximation in a neighbourhood of the classical orbit for the atomic elliptic state of Lena, Delande and Gay expressed in terms of Keplerian coordinates. We hope to investigate these more fully for $\epsilon=\epsilon(u)=\sqrt{\frac{\hbar}{m_{0}}}$, where $m_{0}=m_{0}(u)\sim \infty$ as $u \rightarrow e$, $m_{0}$ so large that $\epsilon(u=e)=O(\epsilon_{0}^{2})$, where $\epsilon_{0}=\sqrt{\frac{\hbar}{m_{00}}}$, $m_{00}$ being the initial mass of the orbiting particle before mass accretion starts.

\subsection {Complex Constants of the Motion}

For the stationary state, $\psi_{n,e}$, the Hamiltonian, $H=2^{-1}\vct{P}^{2}-\mu|\vct{Q}|^{-1}$, is a constant, $H=E$, $E<0$ being the energy. Moreover, if $\tilde{\vct{A}}=(-2E)^{-1/2}\vct{A}$, $\tilde{\vct{A}}$ being the Hamilton–Lenz–Runge vector and $\vct{L}$ is the orbital angular momentum, $\vct{A}$ and $\vct{L}$ are quantum constants of the motion generating the dynamical symmetry group SO(4). Setting the Bohr correspondence limits equal to $\tilde{\vct{a}}=(\tilde{a_{1}},\tilde{a_{2}},\tilde{a_{3}})$, $\boldsymbol \ell=(\ell_{1},\ell_{2},\ell_{3})$, of $\tilde{\vct{A}}$ and $\vct{L}$ defined by the Bohr limits of cartesian coordinates,

$$\tilde{a_{i}}=\lim \psi_{n,e}^{-1} \tilde{A_{i}} \psi_{n,e}, \spc \ell_{i}=\lim \psi_{n,e}^{-1} L_{i} \psi_{n,e}, \spc i=1,2,3,$$

\npgni where $\psi_{n,e}\ne 0$, for $\vct{Z}=\lim(-i\hbar\boldsymbol \nabla\ln\psi_{n,e})$, $\vct{Z}=-i\boldsymbol\nabla R+\boldsymbol\nabla S$, we obtain $\boldsymbol \ell=\vct{r}\wedge \vct{Z}(\vct{r})$, $\vct{a}=\vct{Z}\wedge(\vct{Z}\wedge\vct{r})-\mu r^{-1}\vct{r}$, the semi-classical variables inheriting Pauli’s identities for $\vct{L}$ and $\vct{A}$. So, defining $\tilde{\vct{a}}=(-2E)^{-1}\vct{a}$, we have the following identities:-

\begin{align*}
\ell_{3}+\tilde{a_{1}}\sin\theta&=\lambda \spc(1) \\
\ell_{1}\cos\theta+i\ell_{2}-\tilde{a_{3}}\sin\theta&=0 \spc(2) \\
\tilde{a_{3}}\cos \theta+\ell_{1}\sin \theta&=0 \spc(3) \\
\tilde{a_{1}}\cos \theta+i\tilde{a_{2}}-\ell_{3}\sin \theta&=0 \spc(4)
\end{align*}\vspace{-10mm}

\npgni where $\sin \theta=e$, the eccentricity. (Refs. [16],[17]). Here $E$ has to be the classical energy. Combining equations for $\lambda=\mu(-2E)^{-1/2},$

\begin{align*}
(1)\times \sin \theta+(4)\times \cos \theta \spc \tilde{a_{1}}+i\tilde{a_{2}}\cos \theta=\lambda\sin \theta\spc\textrm{(i)} \\
(2)\times \cos \theta+(3)\times \sin \theta \spc\spc \ell_{1}+i\ell_{2}\cos \theta=0\spc\textrm{(ii)}
\end{align*}\vspace{-10mm}

\npgni So $\ell_{1}=-i\ell_{2}\cos \theta$ giving in (3)\vspace{-5mm}

\begin{align*}
\tilde{a_{3}}-i\ell_{2}\sin \theta=0\spc\textrm{(iii)} \\
\ell_{3}-i\tilde{a_{2}}\sin \theta=\lambda\cos \theta\spc\textrm{(iv)}
\end{align*}\vspace{-10mm}

\npgni Writing $\boldsymbol\ell=\boldsymbol\ell^{\textrm{r}}+i\boldsymbol\ell^{\textrm{i}}$, $\tilde{\vct{a}}=\tilde{\vct{a}}^{\textrm{r}}+i\tilde{\vct{a}}^{\textrm{i}}$, r real part, i imaginary part, assuming non-zero denominators,

$$\frac{\ell_{3}^{\textrm{i}}}{\tilde{a}_{2}^{\textrm{r}}}=-\frac{\tilde{a}_{3}^{\textrm{r}}}{\ell_{2}^{\textrm{i}}}=\frac{\tilde{a}_{3}^{\textrm{i}}}{\ell_{2}^{\textrm{r}}}=e, \spc \frac{\ell_{1}^{\textrm{i}}}{\ell_{2}^{\textrm{r}}}=-\frac{\ell_{1}^{\textrm{r}}}{\ell_{2}^{\textrm{i}}}=\frac{\tilde{a}_{1}^{\textrm{i}}}{\tilde{a}_{2}^{\textrm{r}}}=-\sqrt{1-e^{2}}.$$

\npgni This leaves the real parts of Equations (i) and (iv), giving, where $\psi_{\textrm{sc}}\ne0$,

$$\cos \theta=\frac{\lambda \ell_{3}^{\textrm{r}}+\tilde{a}_{1}^{\textrm{r}}\tilde{a}_{2}^{\textrm{i}}}{(\ell_{3}^{\textrm{r}})^{2}+(\tilde{a}_{1}^{\textrm{r}})^{2}},\spc \sin \theta=\frac{\lambda \tilde{a}_{1}^{\textrm{r}}-\ell_{3}^{\textrm{r}}\tilde{a}_{2}^{\textrm{i}}}{(\ell_{3}^{\textrm{r}})^{2}+(\tilde{a}_{1}^{\textrm{r}})^{2}};\spc \sin \theta=e.$$

\npgni (Ref. [22]).

\subsection {0/0 Limits for the Kepler Problem}

\npgni The last two identities are the generalisations of Newton’s results for planetary motion:

$$e=\sqrt{1+\frac{2\Lambda E}{\mu}},\spc \Lambda=\frac{(\ell_{3}^{\textrm{r}})^{2}}{\mu},$$

\npgni $\Lambda$ being the semi-latus rectum of $\mathcal{E}_K$, the Keplerian ellipse. As $z\rightarrow0$ in approaching the plane of $\mathcal{E}_K$, in a 3-dimensional neighbourhood of $\mathcal{E}_K$, with the exception of $\ell_{3}^{\textrm{r}}$ and $\tilde{a}_{1}^{\textrm{r}}$, the above observables are zero in the limit or are small $O(|\boldsymbol\nabla R|)$ where $|\boldsymbol\nabla R|\rightarrow0$ as we approach $\mathcal{E}_K$.

\npgni In fact $\boldsymbol\ell^{\textrm{i}}=-(\vct{x}\wedge\boldsymbol\nabla R)$, $\boldsymbol\ell^{\textrm{r}}=(\vct{x}\wedge\boldsymbol\nabla S)$, $\vct{a}^{\textrm{i}}=-\boldsymbol\nabla S\wedge(\vct{x}\wedge\boldsymbol\nabla R)-\boldsymbol\nabla R\wedge(\vct{x}\wedge\boldsymbol\nabla S)$,

\npgni $\vct{a}^{\textrm{r}}=-\boldsymbol\nabla S\wedge(\vct{x}\wedge\boldsymbol\nabla S)-\boldsymbol\nabla R\wedge(\vct{x}\wedge\boldsymbol\nabla R)-\dfrac{\mu}{|\vct{x}|}\vct{x}$ giving when $z=0$, $\ell_{1}^{\textrm{i}}=\ell_{2}^{\textrm{i}}=\ell_{1}^{\textrm{r}}=\ell_{2}^{\textrm{r}}=0$ and

$$\frac{\ell_{3}^{\textrm{i}}}{\tilde{a}_{2}^{\textrm{r}}}\rightarrow\frac{\left(y\dfrac{\partial R}{\partial x}-x\dfrac{\partial R}{\partial y}\right)\sqrt{-2E}}{r\left(\dfrac{\partial S}{\partial r}\dfrac{\partial R}{\partial y}+\dfrac{\partial S}{\partial y}\dfrac{\partial R}{\partial r}\right)},$$

$$\frac{\tilde{a}_{1}^{\textrm{i}}}{\tilde{a}_{2}^{\textrm{r}}}\rightarrow\frac{\left(\dfrac{\partial S}{\partial r}\dfrac{\partial R}{\partial r}+\dfrac{\partial R}{\partial r}\dfrac{\partial S}{\partial x}\right)}{\left(\dfrac{\partial S}{\partial r}\dfrac{\partial R}{\partial y}+\dfrac{\partial S}{\partial y}\dfrac{\partial R}{\partial r}\right)}.$$

\npgni Since the above are $C^{1}$ functions of $u$ for $u\in(e,e+\delta e)$ by de l’Hopital\vspace{3mm}

$$\frac{\ell_{3}^{\textrm{i}}}{\tilde{a}_{2}^{\textrm{r}}}\rightarrow\left.\frac{\dfrac{d}{du}\left(y\dfrac{\partial R}{\partial x}-x\dfrac{\partial R}{\partial y}\right)\sqrt{-2E}}{r\dfrac{d}{du}\left(\dfrac{\partial S}{\partial r}\dfrac{\partial R}{\partial y}+\dfrac{\partial S}{\partial y}\dfrac{\partial R}{\partial r}\right)}\right\rceil_{u=0}\spc \textrm{as} \;\; u\rightarrow e,$$ 

\npgni with a similar result for $\dfrac{\tilde{a}_{1}^{\textrm{i}}}{\tilde{a}_{2}^{\textrm{r}}}$. So for 2-dimensional motion in a neighbourhood of $\mathcal{E}_K$, differentiating with respect to u the identities $\tilde{a}_{1}^{\textrm{i}}+\sqrt{1-e^{2}}\tilde{a}_{2}^{\textrm{r}}=0$, $\ell_{3}^{\textrm{i}}-e\tilde{a}_{2}^{\textrm{r}}=0$,

$$\left.\frac{\tilde{a}_{1}^{\textrm{i}\;\prime}}{\tilde{a}_{2}^{\textrm{r}\;\prime}}\right\rceil_{u=e}=-\sqrt{1-e^{2}},\spc \left.\frac{\ell_{3}^{\textrm{i}\;\prime}}{\tilde{a}_{2}^{\textrm{r}\;\prime}}\right\rceil_{u=e}=e\sqrt{-2E},\;\;\forall v,$$

\npgni i.e.\vspace{-5mm}

$$\frac{\tilde{a}_{1}^{\textrm{i}}}{\tilde{a}_{2}^{\textrm{r}}}=-\sqrt{1-e^{2}},\spc \frac{\ell_{3}^{\textrm{i}}}{\tilde{a}_{2}^{\textrm{r}}}=e\sqrt{-2E},\;\;\forall v,$$

\npgni so the left hand sides are constant on $\mathcal{E}_K$ and

$$\sqrt{1-e^{2}}=\frac{\lambda \ell_{3}^{\textrm{r}}+\tilde{a}_{1}^{\textrm{r}}\tilde{a}_{2}^{\textrm{i}}}{(\ell_{3}^{\textrm{r}})^{2}+(\tilde{a}_{1}^{\textrm{r}})^{2}},\spc e=\frac{\lambda \tilde{a}_{1}^{\textrm{r}}-\ell_{3}^{\textrm{r}}\tilde{a}_{2}^{\textrm{i}}}{(\ell_{3}^{\textrm{r}})^{2}+(\tilde{a}_{1}^{\textrm{r}})^{2}}.$$

\npgni These results are inherited from SO(4) symmetry. Note we have proved the constancy of $\dfrac{\ell_{3}^{\textrm{i}}}{\tilde{a}_{2}^{\textrm{r}}}$ and $\dfrac{\tilde{a}_{1}^{\textrm{i}}}{\tilde{a}_{2}^{\textrm{r}}}$ on the Kepler ellipse $\mathcal{E}_K$, as well as the last expressions for $e$ and $\sqrt{1-e^{2}}$, although the last two are just Newton’s results. Similar results can be obtained for Isotropic Oscillators discussed herein.

\section{Conclusion}\vspace{-5mm}

\npgni We have seen in the Bohr correspondence limit the asymptotics of Newtonian gravity for dark matter particles in the atomic elliptic state, $\psi_{n,e}$, for the Kepler problem (and its image under the LCKS transformation - the isotropic harmonic oscillator elliptic state) naturally lead to orbits converging to their classical counterparts. This helps to explain how planets such as Jupiter could be formed in a few million years as well as planetary ring systems such as that of Saturn - something random collisions alone could not explain. Typical of these results is Theorem 3.2 and the preceding arguments. Moreover, we have seen how our isotropic oscillator asymptotics apply to the linearised restricted 3-body problem for Trojan asteroids and their possible formation, whilst presenting two new constants of this motion more generally. There are new results on quantum curvature for this problem given in a forthcoming work (Ref. [28]).\vspace{-3mm}

\npgni Further in this context we have seen how the Bohr correspondence limit of Pauli's equations for the generators of the dynamical symmetry group here, SO(4) and SU(3), lead to 8 identities involving the functions of space $R$ and $S$, where $\psi\sim \textrm{exp}\left(\frac{R+iS}{\epsilon^{2}}\right)$ as $\epsilon^{2}=\frac{\hbar}{m}\rightarrow0$, $m$ being the particle mass. These identities embody complex constants of the motion, involving the real and imaginary parts of the semi-classical angular momentum and Hamilton Lenz Runge vectors. One of these gives an exact analogue in this semi-quantum setting of Newton's results for revolving orbits in Principia (Ref. [3]).

\npgni Of course some of these results reveal quantum corrections to Kepler's laws which could lead to experimental tests of the validity of our ideas. So here we have also given more details of the motion of semi-classical dark matter particles e.g. WIMPs in a neighbourhood of their corresponding classical Keplerian elliptical orbits. For WIMPs carrying some electrical charge in the presence of a magnetic field of a neutron star these results are still relevant. (See forthcoming paper Ref. [27]).

\npgni Moreover, we have investigated how Cauchy's method of characteristics enters into this setting in proving existence and uniqueness theorems for our underlying equations for the functions $R$ and $S$ for stationary states. Here we give just one example of how the quantum particle density limit in the form of $R$ determines the corresponding function $S$. We show how for the zero energy power law this gives the correct limiting orbits in 2 dimensions. Needless to say such results are relevant for all central force problems in this context given angular momentum conservation means the resulting motion ultimately is planar. We plan to give more results of this kind in the future.

\npgni We conclude with the paradigm of semi-classical orbits arising in the Kepler problem - the semi-classical, circular spiral orbit, which we dub the quantum spiral, corresponding to classical circular orbits with radius $a$ and centre $O$, with periodic time $\dfrac{2\pi a^{\frac{3}{2}}}{\mu^{\frac{1}{2}}}$, where $\mu$ is the gravitating point mass at $O$. One never encounters spirals in the classical Kepler problem but the first quantised Kepler problem abounds with them. Could this be a factor in the result that 70\% of galaxies are spiral galaxies and less than 30\% elliptical, spiral galaxies being younger than elliptical ones? Here one would envisage that the proto-galactic nebula of particles would be in a state which is the superposition of orthogonal astronomical elliptic states for different eccentricities and of differing sizes localised in separate parts of the galaxy. The zero eccentricity case is easy to analyse.

\npgni To simplify the presentation we choose unit of length to be $a$ and unit of time to be $\dfrac{2\pi a^{\frac{3}{2}}}{\mu^{\frac{1}{2}}}$. With these assumptions, let $\vct{X}_{t}^{0}=(x_{t},y_{t},z_{t})$ and set $r_{t}=\sqrt{x_{t}^{2}+y_{t}^{2}+z_{t}^{2}}$, $z_{t}=r_{t}\cos(\theta_{t})$ and $\phi_{t}=\tan^{-1}\left(\dfrac{y_{t}}{x_{t}}\right).$ $R$ and $S$ are given explicitly in Ref. [8] and it is proved that the semi-classical equation for $r_{t}$, for $t\ge0$, is

$$\dot{r_{t}}=\frac{1}{r_{t}}-1.$$\vspace{-8mm}

\npgni Given $r_{0}>1$, for $r_{t}>1$, $\dot{r_{t}}<0$, we have

$$r_{t}-1=T(t)\textrm{e}^{-r_{t}},$$\vspace{-8mm}

\npgni where $T(t)=\textrm{e}^{-t}\textrm{e}^{r_{0}}(r_{0}-1)<1$, if $t>r_{0}+\ln(r_{0}-1)$. For sufficiently large $t$, we obtain from Lagrange's expansion

$$r_{t}=1+{\sum_{n=1}^{\infty}}\frac{(T(t))^{n}}{n!}D_{a}^{n-1}(\textrm{e}^{-na})\rceil_{a=1}$$\vspace{-8mm}

\npgni i.e.\vspace{-3mm}

$$r_{t}=1+{\sum_{n=1}^{\infty}}\dfrac{\textrm{e}^{-nt}}{n!}\textrm{e}^{-nr_{0}}(r_{0}-1)^{n}(-n)^{n-1}\textrm{e}^{-n}\rightarrow 1\spc \textrm{as}\spc t\nearrow \infty.$$\vspace{-5mm}

\npgni We also proved in the same reference that

$$z_{t}=r_{t}\cos(\theta_{t})=\frac{r_{t}}{r_{0}}\textrm{exp}\left(-\int_{0}^{t}\frac{ds}{r_{s}^{2}}\right)z_{0}\rightarrow\frac{r_{t}}{r_{0}}z_{0}\textrm{e}^{-t}\rightarrow0\spc \textrm{as}\spc t\nearrow \infty$$\vspace{-8mm}

\npgni and, moreover,\vspace{-5mm}

$$r_{t}^{2}\dot{\phi_{t}}\rightarrow1\spc \textrm{as}\spc t\nearrow \infty.$$\vspace{-10mm}

\npgni This is the simplest example of the spiral orbits converging to the Keplerian ellipses in the infinite time limit with eccentricity zero - see Figure 5 below. The same results obtain for any Keplerian ellipse with eccentricity $e$, $0<e<1$ in our semi-classical treatment.

\begin{center}
\includegraphics[scale=0.75]{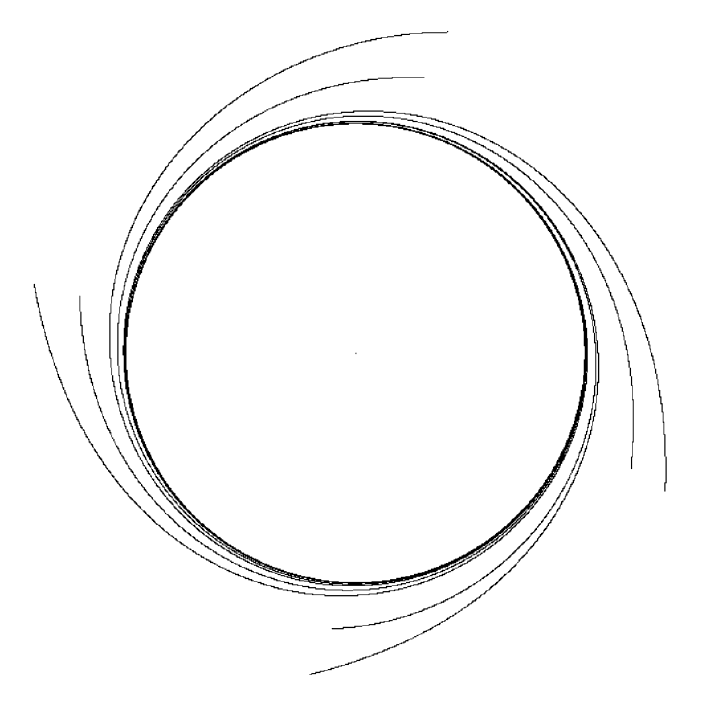}\\
Figure 5.\;:\;\;$\psi_{n,n-1,n-1}(x,y,z)=(x+iy)^{n-1}\textrm{exp}\left(-\dfrac{(x^{2}+y^{2}+z^{2})^{\tph}}{\epsilon^{2}}\right)$.
\end{center}

\npgni The above diagram gives the computer simulation of the transition of the structure of a spiral galaxy with 4 trailing arms to that of an elliptical galaxy with zero eccentricity. This is easily generalised and argues that there is an abundance of dark matter and other similar material in the observable universe behaving in a manner predicted by a first quantisation of the Kepler problem. Hence the present work on its asymptotics. We believe that the first place to look for these is the non-relativistic molecular gases in the arms of spiral galaxies. To test our ideas one needs to measure the quantum curvature and torsion of particle trajectories in the spiral tails of galaxies. For circular Keplerian spirals one can compute the quantum curvature and quantum torsion in 3-dimensions of the orbit $\vct{X}_t^0=(x,y,z)$, viz.

$$\kappa_{q}=\pm\frac{a^2(2(2z^2r^2+\rho^4)z^2\rho^2+(2ar^3-\rho^4)^2)^\frac{1}{2}}{(2a^2r^2-2a\rho^2r+\rho^4+\rho^2z^2)^\frac{3}{2}\rho},$$

$$\tau_{q}=\frac{(3\rho^8-a(\rho^4+z^4)r^3)z}{2(2z^2r^2+\rho^4)z^2\rho^2r^2+(2ar^3-\rho^4)^2r^2},$$

\npgni where $r^2=x^2+y^2+z^2$, $\rho^2=x^2+y^2$ and $a=\dfrac{\lambda^2}{\mu}$ is the radius of the corresponding classical circular orbit.\vspace{-3mm}

\npgni These tests may prove unrealistic so we give two simpler tests - Anti-gravity and Spiral Splitting.\vspace{-2mm}

\npgni \textbf{Anti-gravity}\vspace{-3mm}

\npgni On the ellipse at infinity, with equation, $u=-e$, for any eccentricity $e$, $0<e<1$, $b_{u}=b_{v}=0$, and our semi-classical cloud particles have velocities $\sim\dfrac{a}{T}$ , $T$ being the periodic time for the ellipse with eccentricity $e$, $a$ the semi-major axis. A simple calculation shows the effective potential, $V_{\textrm{eff}}\sim\dfrac{\mu}{r}$ as $r\sim\infty$, corresponding to anti-gravity effects for any $e$. Indeed it is easy to check this for the circular spiral above, in the plane $z=0$, for which\vspace{-3mm}

$$V_{\textrm{eff}}=\frac{\mu}{r}-\frac{\lambda^2}{r^2}-\frac{\mu^2}{\lambda^2},\;\;\;r>0, \;\;\;V_{\textrm{eff}}'(2a)=0.$$\vspace{-8mm}

\npgni The figure below compares the graphs of $V$ (black line) and $V_{\textrm{eff}}$ (red line).\vspace{3mm}

\begin{center}
\includegraphics[scale=0.7]{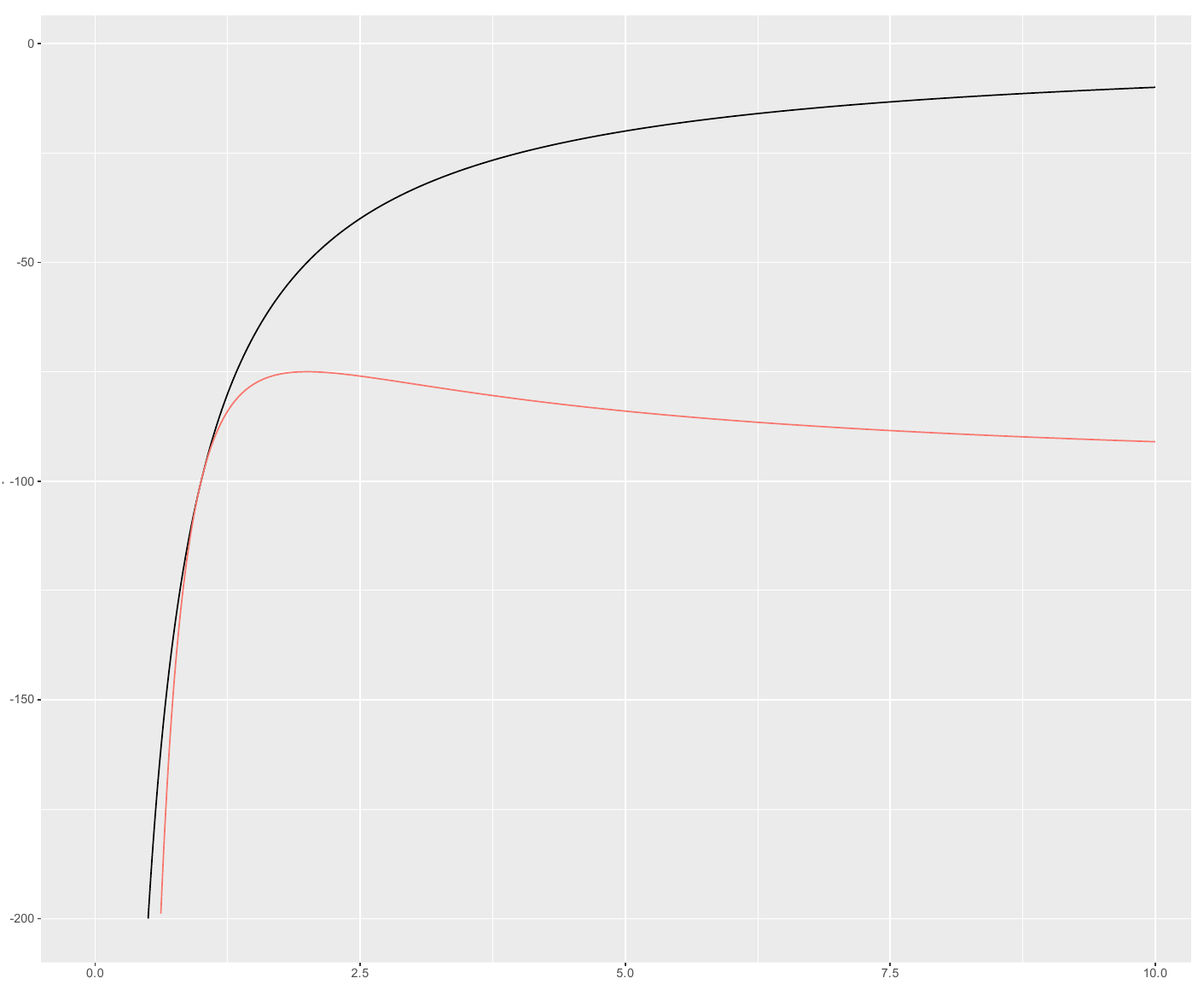}\\
Figure 6.
\end{center}

\npgni So in this case WIMPish particles slow down as they approach the circle $r=2a$, after which they speed up. A similar result holds for small eccentricity $e$.

\npgni Of course when $e>1$ the corresponding hyperbolic orbits ($E>0$) could be mistaken for anti-gravity effects and a repulsive force between inertial frames embedded in galaxies.

\npgni \textbf{Superposition and Spiral Splitting}

\npgni Superposing two stationary state solutions, for possibly different energies, of the semi-classical equations, $\phi=\phi_{1}+\phi_{2}$, where

$$\phi_{j}\sim\textrm{exp}\left(\frac{R_{j}+iS_{j}}{\epsilon^2}\right),\;\;j=1,2\;\;\textrm{and}\;\; \phi\sim\textrm{exp}\left(\frac{R+iS}{\epsilon^2}\right)\;\;\textrm{as}\;\;\epsilon\sim0$$\vspace{-8mm}

\npgni leads to\vspace{-5mm}

$$R\sim\dfrac{\rho_{1}}{\rho_{1}+\rho_{2}}R_{1}+\dfrac{\rho_{2}}{\rho_{1}+\rho_{2}}R_{2},\;\;\rho_{j}=\textrm{exp}\left(\dfrac{2R_{j}}{\epsilon^2}\right),\;\;j=1,2.$$

\npgni So the leading behaviour of the asymptotics gives rise to the formal result,

$$R\sim R_{\textrm{M}},\;\;\;R_{\textrm{M}}(r)=\textrm{max}(R_{1}(r),R_{2}(r)),$$\vspace{-8mm}

\npgni $S\cong S_{\textrm{M}}$,  $R+iS\cong R_{\textrm{M}}+iS_{\textrm{M}}$,  $\boldsymbol{\nabla}R\cong R_{\textrm{M}}$,  $\boldsymbol{\nabla}S\cong S_{\textrm{M}}$, in the formal limit $\epsilon\sim0$. (See Kolokoltsov and Maslov - Ref. [11]).  What does this tell us about the formation of ring systems in celestial mechanics if we take it seriously for two circular spiral states?

\npgni Approximating $R_{\textrm{j}}$ by Taylor's theorem for radii $a_{j}$, $j=1,2$, since

\npgni $R\sim-\sqrt{\dfrac{\mu}{a}}\left(r-a-a\ln{\dfrac{r}{a}}\right)$, so $\rho_{j}(r)\sim\textrm{exp}\left(\dfrac{R_{j}''(a_{j})}{\epsilon^2}(r-a_{j})^2\right)$ as $\epsilon\sim0$, giving\vspace{3mm}

$$\rho_{j}(r)\sim\textrm{exp}\left(-\sqrt{\dfrac{\mu}{a_{j}^3}}\dfrac{(r-a_{j})^2}{\epsilon^2}\right)\;\;\textrm{as}\;\;\epsilon\sim0,$$

\npgni for our circular spirals. So, if we ask when is $R_{1}(r)=R_{2}(r)$, we obtain the two circles $r=r_{\pm}(a_{1},a_{2})$, where\vspace{-2mm}

$$\dfrac{r-a_{1}}{a_{2}-r}=\pm\left(\dfrac{a_{1}}{a_{2}}\right)^{\frac{3}{4}},\;\;\;0<r_{-}<a_{1}<r_{+}<a_{2}<\infty.$$\vspace{-6mm}

\npgni This means that in this rather crude limit we have spiral splitting on the circles $r=r_{\pm}$, which of course will be most pronounced for $r=r_{+}$, the angle of splitting, $\theta_{12}$, being given by\vspace{-3mm}

$$\cos{\theta_{12}}=\widehat{\boldsymbol{\nabla}(R_{1}+S_{1})}.\widehat{\boldsymbol{\nabla}(R_{2}+S_{2})}\rceil_{r=r_{\pm}}.$$\vspace{-10mm}

\npgni Hence we have a rough and ready test of our ideas for ring systems being formed. Moreover, this gives the semi-classical limit of the wave function for planetary systems such as our own.

\npgni Let $\phi=\sum\limits_{j=1}^{n}\chi_{_{(R_{M}=R_{j})}}c_{j}\phi_{j}$, $\phi_{j}$ corresponding to coplanar circular spiral orbits, with radii $a_{j}$, $a_{1},a_{2},\cdots,a_{n}$, arranged in ascending order, normalised so that $|\phi_{j}(a_{j})|=1$ and $c_{j}>0$, with $c_{j}^2$ the relative probability of the circular WIMP ending in the $j^{\textrm{th}}$ ring. We know that for sufficiently small $|r-a_{i}|$, $R_{\textrm{M}}(r)=R_{i}(r)$. We denote by $\bigtriangleup^{+}r_{i}$, $r=(a_{i}+\bigtriangleup^{+}r_{i})$, $\bigtriangleup^{+}r_{i}>0$, where $|r-a_{i}|$ is the least value for which the $i$ in the equation $R_{\textrm{M}}(r)=R_{i}(r)$ changes i.e. $R_{i}(r)=R_{k}(r)$, for some $k$, $0<k\le n$.\vspace{5mm}

\begin{corollary}

Potentially $\bigtriangleup^{+}r_{i}=\min\limits_{i\ne j}\left(\dfrac{\pm\left(\dfrac{a_{i}}{a_{j}}\right)^{\frac{3}{4}}}{1\pm\left(\dfrac{a_{i}}{a_{j}}\right)^{\frac{3}{4}}}(a_{j}-a_{i})\right)$, but only the

\npgni plus signs are relevant and then only those for which the minimum is positive.
    
\end{corollary}

\begin{proof}

Only the plus sign gives $\bigtriangleup^{+}r_{i}>0$ as we see. Setting $\alpha=\left(\dfrac{a_{i}}{a_{j}}\right)^{\frac{3}{4}}$, for $j>i$, $0<\alpha<1$; and for $j<i$, $\alpha>1.$ So observing that for the minus sign, for $j>i$, $a_{j}-a_{i}>0$

$$-\dfrac{\alpha}{1-\alpha}(a_{j}-a_{i})=\dfrac{\alpha}{\alpha-1}(a_{j}-a_{i})<0$$

\npgni and for $j<i$, $a_{j}-a_{i}<0$ and $\dfrac{\alpha}{1-\alpha}<0$.
    
\end{proof}

\npgni Similarly we define $\bigtriangleup^{-}r_{i}$ by $r=r_{i}-\bigtriangleup^{-}r_{i}$, $\bigtriangleup^{-}r_{i}>0$, where $R_{\textrm{M}}(r)=R_{i}(r)$ switches.

\begin{corollary}

$$\bigtriangleup^{-}r_{i}=\min\limits_{i\ne j}\left|\dfrac{\left(\dfrac{a_{i}}{a_{j}}\right)^{\frac{3}{4}}}{1-\left(\dfrac{a_{i}}{a_{j}}\right)^{\frac{3}{4}}}(a_{j}-a_{i})\right|.$$
    
\end{corollary}

\npgni This enables us to reverse engineer the wave function for coplanar circular components of nebulae which go on to form parts of solar systems e.g. for our own solar system - Venus, Earth, Mars, Jupiter, Saturn, Uranus and Neptune.

\npgni So far we did not discuss in detail the requirements that the $S$ term in our equations off the classical orbit be a gradient (c.f. $(R,V)$ equation on page 29). A careful reading of Nelson's original ideas on stochastic mechanics tells us that when this fails there has to be some underlying motion of the aether invalidating the stochastic version of Newton's $2^{\textrm{nd}}$ Law encapsulated in the Schrödinger formulation. In our formulation at the corresponding stage (leaving aside the possibility of vector potentials) we can only conclude that there is no stationary state satisfying our requirements. Given the astrnomical length scales over which we are testing the asymptotics of this Schrödinger picture it is not surprising if this occurs here sometimes. What is truly amazing to our minds is that the asymptotics of the Schrödinger wavefunctions for states considered herein (especially the astronomical elliptic states of Lena, Delande and Gay) seem capable of explaining not only the rapid formation of ring systems for massive planets and solar systems themselves but also the evolution of some galaxies. These are the laboratories in which our ideas need to be tested, as our future work will show.

\npgni Finally we include a diagram resulting from the analysis at the end of section (5) of a Coulomb, harmonic oscillator and dark energy potential inspired by Whittaker (Ref. [29]). (See also Ref. [25]).

\begin{center}
\includegraphics[scale=0.75]{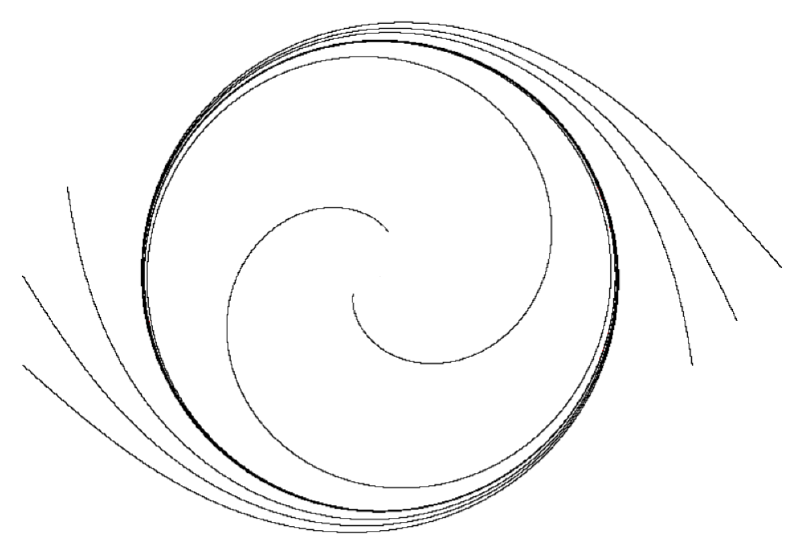}\\
Figure 6.
\end{center}

\npgni All our results beg the question as to how our work fits in the broader context of mathematical physics and general relativity, in particular in connection with the effects of Jupiter and Mach's Principle. Both of these are discussed in Ref. [28] where we give a generalisation of Lagrange's equilateral triangle solution for the 3-body problem, new results for the Foucault pendulum in the classical setting and how these have ramifications for the quantum case. In the context of our and Newton's absolute space the (HJI - Hamilton-Jacobi-Integral) equation (pp. 17) resolves any problems we might have had from Mach's Principle. In spite of the fact that 2027 is the $300^{\textrm{th}}$ anniversary of Newton's death we believe that still Newton lives.

\section*{Dedication and Thanks}\vspace{-8mm}

\npgni This work is dedicated to Sir Michael Atiyah, a personal mentor and friend of one of us and an inspiration to us all, he is so sorely missed. It is also a pleasure to thank our teachers G.B.S, G.R.S., J.C.T. and D.W.\vspace{-3mm}

\npgni On completion of this work we received a copy of the following paper, containing a more complete set of references, which is also relevant to Nelson's stochastic mechanics:\vspace{-3mm}

\npgni Cresson, J., Nottale, L., and Lehner, T. “Stochastic modification of Newtonian dynamics
and induced potential - Application to spiral
galaxies and the dark potential” J. Math. Phys. 62, 072702 (2021).

\end{document}